\documentclass[11pt,tightenlines,eqsecnum,floats,aps,amsmath,superscriptaddress,amssymb,nofootinbib,longbibliography]{revtex4}
\usepackage{graphicx, wrapfig}
\usepackage{amssymb}
\usepackage[usenames, dvipsnames]{color}
\usepackage{diagbox}
\usepackage{mathrsfs}
\usepackage{subfigure}
\usepackage{float}
\usepackage{slashbox}
\usepackage{verbatim}
\usepackage{cancel}
\usepackage[normalem]{ulem}
\usepackage{amsthm}
\usepackage{enumerate}
\usepackage[hidelinks]{hyperref}
\hypersetup{
	colorlinks   = true, 
	urlcolor     = blue, 
	linkcolor    = blue, 
	citecolor   = blue 
}
\newtheorem{theorem}{Theorem}

\newcommand{\be}{\nopagebreak[3]\begin{equation}}
	\newcommand{\ee}{\end{equation}}
\newcommand{\bfig}{\nopagebreak[3]\begin{figure}}
	\newcommand{\efig}{\end{figure}}
\newcommand{\bea}{\nopagebreak[3]\begin{eqnarray}}
	\newcommand{\eea}{\end{eqnarray}}

\newcommand{\bmult}{\nopagebreak[3]\begin{multline}}
	\newcommand{\emult}{\end{multline}}

\begin{document}
	
	\title{Does inflation squeeze cosmological perturbations?}
	
	\author{Ivan Agullo}
	\email{agullo@lsu.edu}
	\affiliation{Department of Physics and Astronomy, Louisiana State University, Baton Rouge, LA 70803, U.S.A.}
	\author{B\'eatrice Bonga}
	\email{bbonga@science.ru.nl}
	\author{Patricia Ribes Metidieri}
	\email{patricia.ribesmetidieri@ru.nl}
	\affiliation{Institute for Mathematics, Astrophysics and Particle Physics,
		Radboud University, 6525 AJ Nijmegen, The Netherlands}
	
	\begin{abstract}

		There seems to exist agreement about the fact that  inflation squeezes the quantum state of cosmological perturbations and entangles modes with wavenumbers $\vec k$ and $-\vec k$. Paradoxically, this result has been used to justify both the classicality as well as  the quantumness of the primordial  perturbations at the end of inflation.  We reexamine this question and point out that the
		definition of two-mode squeezing of the  modes $\vec k$ and $-\vec k$ used in  previous work rests on choices that are only justified for systems with  time-independent Hamiltonians and finitely many degrees of freedom. We argue that for quantum fields propagating on generic time-dependent  Friedmann-Lema\^itre-Robertson-Walker backgrounds, the notion of squeezed states is subject to ambiguities, which go hand in hand with the ambiguity in the definition of particles. In other words, we argue that the question ``does the cosmic expansion squeeze and entangle modes with wavenumbers $\vec k$ and $-\vec k$?'' contains the same ambiguity as the question ``does the cosmic expansion create particles?''. 
		When additional symmetries are present, like in  the (quasi) de Sitter-like spacetimes used in inflationary models, one can resolve the ambiguities, and we find that the answer to the question in the title turns out to be in the negative.  We further argue that this fact does not make the state of cosmological perturbations any less quantum, at least when deviations from Gaussianity can be neglected. 
		
	\end{abstract}

	\maketitle
	
	
	\section{Introduction\label{intro}}
	
	The paradigm of cosmic inflation gave rise to an unforeseen and profound lesson: the density perturbations in the universe 
	may have a {\em quantum origin} \cite{Mukhanov:1981xt,Hawking:1982cz,Guth:1982ec,Starobinsky:1982ee,Bardeen:1983qw}. In inflation, density perturbations originate from the quantum fluctuations of the vacuum itself, which were amplified and stretched to cosmological distances by the accelerated cosmic expansion. 
	This claim is of indisputable conceptual depth and beauty, and many efforts have been dedicated to investigate it \cite{Grishchuk:1990bj,Albrecht:1992kf,Polarski:1995jg,Lesgourgues:1996jc,Kiefer:1998qe,Kiefer:1998pb,Kiefer:2008ku,sudarsky_shortcomings_2011,Martin:2015qta,ack,Brahma:2021mng,Green:2020whw}. In this paper, we further investigate this fundamental question: is there a way to confirm or refute the genuine quantum origin of the cosmic perturbations? 
	
	
	At present, cosmologists use purely classical tools to analyze the collected  data, and there is no evidence so far that such tools are insufficient to understand observations. More concretely, in contrasting the predictions of inflation with observations, one replaces the quantum probability distribution for the primordial perturbations ---computed using quantum field theory--- with a classical stochastic function with identical statistical moments. In doing so, one automatically eliminates any genuinely quantum trace.  This situation has motivated different researchers to investigate two natural questions: (i) If perturbations have a quantum origin, how can we understand the apparent classicality of our universe? (ii) Is there any observable in the cosmic microwave background (CMB) which could prove  that a classical treatment is insufficient?
	
	Paradoxically, a single mechanism has been in the spotlight of the search for an answer to these two
	questions: dynamical generation of two-mode squeezing during inflation between  perturbations with wavenumbers $\vec k$ and $-\vec k$. On the one hand, it has been argued that  this squeezing mitigates many quantum aspects of the perturbations \cite{Grishchuk:1990bj,Albrecht:1992kf,Lesgourgues:1996jc,Kiefer:1998pb,Kiefer:1998qe,Kiefer:2008ku,Polarski:1995jg} (see \cite{Hsiang:2021kgh} for a recent criticism to these arguments) while, on the other hand, it has also been argued that squeezing comes together with a  generation of quantum entanglement between the modes $\vec k$ and $-\vec k$, which makes the state of perturbations at the end of inflation very quantum \cite{Martin:2015qta}. \emph{Single}-mode squeezing and its relation to the quantumness of the state of perturbations during inflation has also been discussed in \cite{ack}. 
	
	The aim of this paper is to take a critical view on the definition of two-mode (and also single-mode) squeezing for quantum fields on Friedmann-Lema\^itre-Robertson-Walker (FLRW) spacetimes.
	Following earlier work, we will focus on Gaussian states, since observations have not revealed any sign of primordial non-Gaussianity, despite important efforts \cite{refId0} (see \cite{Green:2020whw,Shandera:2017qkg,Brahma:2021mng} for discussions of non-Gaussian states). 
	
	Our main goal is to point out an ambiguity underlying most discussions on the generation of squeezing in Fourier space by the cosmic expansion. In talking about squeezing and entanglement between degrees of freedom  with the wavenumbers  $\vec k$ and $-\vec k$, one needs to construct canonically conjugated pairs of Hermitian operators associated with these degrees of freedom. We discuss the ambiguities one finds in this construction when the underlying spacetime is homogeneous but  time-dependent, and argue they are the same ambiguities one finds in the definition of vacuum or particle. 
	In a generic FLRW, there is no preferred choice, and therefore the answer to this question does not carry any invariant physical meaning.   We also argue that the answer is not of direct relevance to understand observations ---which are carried out in {\em real} space--- as one would expect given the inherent ambiguity.

	Entanglement in real space is ubiquitous in quantum field theory \cite{1985PhLA..110..257S,2018}, even for the vacuum in Minkowski spacetime, and it is independent of any particle interpretation. But current cosmological data seem insufficient to reveal any trace of this entanglement, due to the difficulty in observing non-commuting observables associated with the primordial perturbations. It is for this reason that a classical stochastic state suffices to completely account for observations.   
	Along the way, we will use simple examples to illustrate the main messages of this paper, using a set of finitely many harmonics oscillators and a linear scalar field in FLRW spacetimes.

	
	This paper is organized as follows. We begin in section \ref{sec:findim} with a brief review of squeezing and entanglement for quantum systems with finitely many degrees of freedom and quadratic Hamiltonians. We summarize the relation between squeezing, entanglement, and ``quantumness'' of Gaussian states. The case of a time-dependent Hamiltonian serves to illustrate several messages which will be  important for the study of squeezing in inflation. Although the lessons extracted from this simple analysis are not new, they are not made explicit in many treatments. In section~\ref{sec:scalar}, we extend the discussion to field theory by considering a scalar field in spatially flat FLRW spacetimes. We discuss the additional subtleties that the existence of infinitely many degrees of freedom introduces. To illustrate the role of symmetries in the dynamical  generation of squeezing, we consider the example of the Poincar\'e patch of de Sitter spacetime, and use it to compare with the strategy followed in earlier work. This example provides lessons of direct applicability for cosmological perturbations in inflation, which are discussed in section~\ref{sec:cosmpert}. We collect our results and put them in a broader perspective in section \ref{sec:concl}. Appendix~\ref{app:single-mode-sqz} contains a discussion of single-mode squeezing during inflation, appendix~\ref{app:qq} summarizes other measures of ``quantumness'' commonly use in the literature of quantum optics, such as the $P$-function, and appendix~\ref{app:BD} provides a derivation of the Bunch-Davies vacuum in the Poincar\'e patch of de Sitter spacetime and its properties in the Schr\"odinger evolution picture, many of which are used in the main text. Throughout this paper, we use units in which $\hbar=c=1$. 
	

	\section{Squeezing, entanglement and quantumness of Gaussian states of linear finite-dimensional systems\label{sec:findim}}
	
	The goal of this section is to emphasize three messages concerning  finite-dimensional bosonic systems: 
	(i) The notion of squeezing requires a quantum state {\em and} a pair of non-commuting operators. 
	(ii) For any  Gaussian state there always exists a  basis of canonically conjugated pairs of operators for which the state is not squeezed, and another basis for which the state has arbitrarily large squeezing. The same applies to entanglement: one can always find bi-partitions of the system for which the entanglement between the two sub-systems is zero, or as large as desired. Hence, the sentence ``$\hat \rho$ is a squeezed or an entangled state'' is empty, unless one has in mind a preferred set of canonically conjugated pairs or bi-partition. 
	(iii) If the Hamiltonian is time-dependent, the preferred canonically conjugated pairs and bi-partitions at the initial and final instants are generically different. 
	Therefore, the question ``does evolution squeeze or entangle the state $\hat \rho$?'' brings an additional ambiguity, related to the choice of quadrature-pairs and bi-partitions at the initial {\em and} final times. 
	
	In the remainder of this section, we justify these statements and illustrate them with simple examples. Most of this material is known (see for instance \cite{Zanardi:2004zz}), and our goal is to simply emphasize  aspects that are frequently unnoticed, and that are relevant  for the questions investigated in  this article.
	The reader familiar with these topics can jump directly to the next section. 
	
	\subsection{Linear finite-dimensional systems: basic notation}
	
	We begin by introducing standard terminology for the quantization of a system with a $2N$-dimensional phase space, for finite $N$. We will focus on linear systems for which the classical phase space $\Gamma$ is a vector space, and  one can choose global canonical coordinates $x_I,p_I$, $I=1,\ldots, N$ in $\Gamma$ such that the  Poisson brackets are $\{x_I,x_J\}=\{p_I,p_J\}=0$ and $\{x_I,p_J\}=\delta_{IJ}$. This can be expressed more compactly by defining the column vector $r^i=(x_1,p_1,\ldots,x_N,p_N)^{\top}$ ---we use lower case letters for indices in phase space, $i=1,\ldots, 2N$---  in terms of which  all  Poisson brackets read $\{r^i,r^j\}=\Omega^{ij}$, where $\Omega^{ij}=\oplus_N \begin{pmatrix} 0 & 1 \\-1 & 0\end{pmatrix}$ is the (inverse of) the symplectic structure. 
	
	In the quantum theory, the canonical coordinates $r^i$ are promoted to operators satisfying commutation relations $[\hat r^i,\hat r^j]=i\, \Omega^{ij}$. Together with the identity operator $\hat{\mathbb{I}}$,  $\hat r^i$ can be used to generate all other polynomial operators by taking linear combinations of their products. Of particular relevance for our discussion are {\em linear} observables, made of simple linear combinations of  $\hat r^i$:
	\be 
	\hat O_{\vec \alpha}\equiv \vec \alpha\cdot \hat{\vec r}=\alpha_i \, \hat r^i\, , \qquad \text{with} \quad \vec \alpha\in \mathbb{R}^{2N} 
	\ee
	(sum over repeated indices is understood). It is convenient to identify $\vec \alpha$ with elements of $\Gamma^*$, the dual of the phase space $\Gamma$, since in that way all operators $\hat O_{\vec \alpha}$ have dimensions of action. 
	Given two linear observables, $ \hat O_{\vec \alpha}$ and  $\hat O_{\vec \beta}$, their commutator is simply $[\hat O_{\vec \alpha},\hat O_{\vec \beta}]=i\, \alpha_i\beta_j\Omega^{ij}$, or in matrix notation $[\hat O_{\vec \alpha},\hat O_{\vec \beta}]=i\, \vec \alpha^{\top} \cdot \Omega \cdot \vec \beta$ (that is, $i$ times the symplectic product of $\vec \alpha$ and $\vec \beta$). 
	
	We  say that two such operators $ \hat O_{\vec \alpha}$ and  $\hat O_{\vec \beta}$ form a quadrature-pair if $[\hat O_{\vec \alpha},\hat O_{\vec \beta}]=i$. Furthermore, $N$ mutually commuting quadrature-pairs 
	will be said to form a Darboux basis. For instance, the canonical operators $\hat r^i$ ---properly normalized, so they all have dimensions of action while still satisfying the same commutation relations--- form a Darboux basis of quadrature-pairs.  Given such a basis, and given a real $2N\times 2N$ matrix $S^i_{\, j}$ that leaves the symplectic structure invariant, i.e.,  satisfying $S^{\top}\cdot \Omega\cdot S=\Omega$, the linear operators $ \hat{\vec r}\, '=S\cdot \hat{\vec r}$  also form a Darboux basis of quadrature-pairs. The matrix $S$ implements a linear canonical transformation, and the set of all such matrices forms the symplectic group $\rm{Sp}(2N,\mathbb{R})$. 
	
	\subsection{Quadrature squeezing\label{quadsqz}}
	We call  a quantum state (Gaussian or not, pure or mixed) $\hat \rho$
	squeezed relative to the quadrature-pair  $(\hat O_{\vec \alpha},\hat O_{\vec \beta})$  
	when either of the dispersions  
	$ \Delta O^2_{\vec \alpha}$ or $\Delta O^2_{\vec \beta}$  satisfy
	\be  \Delta O^2_{\vec \alpha}< \frac{1}{2}\,   \qquad {\rm or} \qquad  \Delta O^2_{\vec \beta}< \frac{1}{2} \,, \ee
	%
	where $\Delta O^2\equiv {\rm Tr}[\hat \rho \, \hat O^2]-{\rm Tr}[\hat \rho \,\hat O]^2$. (Of course,  Heisenberg's principle implies that the product $\Delta O^2_{\vec \alpha}\, \Delta O^2_{\vec \beta}$ is never less than $\frac{1}{4}$.) 
	
	Note that in the definition of squeezing it is pivotal that both $\hat O_{\vec \alpha}$ and  $\hat O_{\vec \beta}$  have the same dimensions, otherwise there is no unambiguous way of splitting Heisenberg's uncertainty lower bound between them. It is also important to emphasize that the notion of squeezing requires both a state \emph{and} a pair of quadratures. It is meaningless to simply say that a state is squeezed. 
	
	Given two commuting quadrature-pairs, ($\hat O_{\vec \alpha_1}, \,  \hat O_{\vec \beta_1})$ and ($\hat O_{\vec \alpha_2}, \,  \hat O_{\vec \beta_2})$, each describing a physical degree of freedom (or mode) of the system, one says the state $\hat \rho$ is a  {\em two-mode} squeezed state relative to these  pairs when it is squeezed for  any non-trivial  linear combination of the two pairs. 
	%
	%

	\subsection{Gaussian states and squeezing}

	We focus now on Gaussian states. This is the family of states most discussions of squeezing and classicality in cosmology have focused on, motivated by the absence of primordial non-Gaussianity in the CMB (see, however, \cite{Green:2020whw,Shandera:2017qkg,Brahma:2021mng}). 
	The proofs omitted in this section can be found, for instance, in \cite{serafini2017quantum}. 
	The simplicity of quantum Gaussian states resides in the fact that all their quantum moments $\langle \hat r^{i_1}\cdots  \hat r^{i_{n}}\rangle$ are completely determined from the first  and second moments, $\langle \hat r^{i}\rangle$ and $\langle \hat r^{i}\hat r^{j}\rangle$, respectively. We will denote the first moments by $\mu^i=\langle \hat r^{i}\rangle$. The non-trivial information in the second moments is more cleanly encapsulated in their symmetrized version $\sigma^{ij}=\langle \{(\hat r^i-\mu^i), (\hat r^j-\mu^j)\}\rangle$, where the curly brackets indicate the symmetric anti-commutator.\footnote{The anti-symmetric part of the second moments is proportional to the commutator $[\hat r^i, \hat r^j]=i\, \Omega^{ij}$, whose expectation value does not carry any information about the state.} The matrix $\sigma^{ij}$ is called the covariance matrix of the state, and it carries information about the dispersion of all linear operators $\hat O_{\vec \alpha}\equiv \vec \alpha\cdot\hat{\vec r}$:
	\be \Delta O^2_{\vec \alpha} =\frac{1}{2} \, \vec \alpha^{\top}  \cdot \sigma\cdot \vec \alpha=\frac{1}{2}\alpha_i\alpha_j\sigma^{ij}\, .\ee
	Therefore, all physical predictions relative to a Gaussian state, pure or mixed, can be obtained from the $2N$-dimensional vector $\mu^i$, and the $2N\times2N$-symmetric matrix $\sigma^{ij}$. 
	
	An important property of the covariance matrix $\sigma$ is the following. First, it must satisfy that  $\sigma+i\, \Omega$ is a non-negative matrix (i.e., all eigenvalues must be non-negative). This condition is tantamount to the positivity of the density matrix and it encodes the quantum uncertainty inequalities,  in particular Heisenberg's principle. This further implies that  $\sigma$ is a positive-definite matrix.
	Williamson's theorem then guarantees that $\sigma$ can be ``symplecticly diagonalized'', i.e., that there exists a symplectic transformation $B \in \rm{Sp}(2N,\mathbb{R})$ such that  $B\cdot \sigma \cdot B^{\top}=\rm{diag}(\nu_1,\nu_1,\cdots, \nu_N,\nu_N)$. The $N$ positive real numbers $\nu_I$ are called the symplectic eigenvalues of $\sigma$, and they encode the invariant information of the covariance matrix.\footnote{Invariant in the sense of the symplectic group in phase space. Note that, in contrast, the ordinary eigenvalues of $\sigma$ are not invariant under symplectic transformations;  
		they depend on both, the quantum state and the Darboux basis chosen to write  $\sigma$.\label{fn1}} 
	
	One can further prove that all symplectic eigenvalues have to be larger than one, $\nu_I\geq 1$---and they are all equal to one if and only if the state is pure. 
	The symplectic eigenvalues $\nu_I$, by construction, are equal to twice the dispersion of the operators in the Darboux basis in which the covariance matrix takes the diagonal form, that is, $\Delta X_I^2=\Delta P_I^2=\frac{\nu_I}{2}$,  $I=1,\cdots, N$. 
	Hence, none of the operators in this basis are squeezed. This proves that for any Gaussian state there always exists a Darboux basis of canonical operators for which the state is manifestly not squeezed. 
	
	On the contrary, one can always  find quadrature-pairs with an arbitrary large degree of squeezing. A trivial example is the following: if  ($\hat O_{\vec \alpha}$, $\hat O_{\vec \beta}$) is a quadrature-pair, the transformation $\hat O_{\vec \alpha}\to e^{-\zeta}\, O_{\vec \alpha} $ and $\hat O_{\vec \alpha}\to e^{\zeta}\, O_{\vec \alpha} $, parameterized by a real number $\zeta$, is a symplectic transformation which brings us to a  new quadrature-pair with  dispersions  $e^{-2\, \zeta}\,  \langle \Delta \hat O_{\vec \alpha}^2\rangle$ and $e^{2\, \zeta} \, \langle \Delta \hat O_{\vec \beta}^2\rangle$, respectively. Hence, tuning $\zeta$  one can obtain quadratures with as much or as little squeezing as desired. Therefore, the statement ``$\hat \rho$ is a  squeezed state'' is empty---unless one has in mind some preferred quadratures.

	In practical cases, one can use  physical arguments to choose some preferred quadratures among the observables that are accessible in an experiment. For instance, for a time-independent harmonic oscillator of mass $m$ and frequency $w$, $ \hat X\equiv (m\, w)^{1/2}\, \hat x$ and $\hat P\equiv  (m\, w)^{-1/2}\, \hat p$ is the most physically relevant quadrature-pair, 
	and it is natural to associate squeezing with this pair. 
	%
	%
	For more general but time-independent systems with quadratic Hamiltonians, the normal modes provide a physically preferred Darboux basis of quadrature-pairs to which one can naturally refer when talking about squeezing. The process is as follows. The ground state of the Hamiltonian defines a set of annihilation operators $\hat A_I$ 
	from which we can write the normal modes of the system as $\hat X_I=\frac{1}{\sqrt{2}} \, (\hat A_I+\hat A_I^{\dagger})$ and $\hat P_I=\frac{-i}{\sqrt{2}} \, (\hat A_I-\hat A_I^{\dagger})$. It is then straightforward to prove that the ground state, and in fact all eigenstates of $\hat A_I$ (i.e., coherent states), are not squeezed relative to the basis of quadrature-pairs made of the normal modes $(\hat X_I,\hat P_I)$. 
	The reason we repeat this well-known fact  is to emphasize that  this familiar notion of squeezed  states rests on the existence of a time-independent Hamiltonian, which has a ground state and an associated Darboux basis of quadrature-pairs made of normal modes. These elements  will not be available for fields in time-dependent geometries, for which the difference between squeezed and coherent states dilutes, and become subject to certain choices.  
	

	\subsection{Gaussian states and entanglement}
	
	The conclusions of the previous subsection also apply to entanglement between bi-partitions of any $N$-mode systems and Gaussian states. Entanglement is  not an invariant notion that can be attributed only to the quantum state, since it obviously depends on the  bi-partition chosen \cite{Zanardi:2004zz}. The arguments summarized in the previous section ---in particular, the fact that the covariance matrix can always be brought to diagonal form--- automatically imply that, given an arbitrary Gaussian state, pure or mixed, there always exist a Darboux basis of quadrature-pairs, $\hat X_1, \hat P_1,\cdots, \hat X_N,\hat P_N$, for which there is no entanglement among any bi-partition of these pairs, and the state is in a manifestly non-entangled and uncorrelated  form. On the contrary, one can always build suitable linear combinations of these operators to find bi-partitions for which the entanglement is arbitrarily high.
	For instance, the subsystems defined by the canonical pairs $(\hat{ X}'_1\equiv \hat X_1,\, \hat{ P}'_1\equiv\hat P_1+\hat X_2) $ and $(\hat{ X}'_2\equiv\hat X_2,\, \hat{P}'_2\equiv\hat P_2+\hat X_1)$ can be entangled for states for which the  subsystems $(\hat X_1,\hat P_1) $ and $(\hat X_2,\,\hat P_2)$  are not, as it can be checked for two harmonic oscillators in the vacuum state. Hence, the sentence ``the Gaussian state $\hat \rho $ is an entangled state" is incomplete, unless one specifies a bi-partition. Consequently, little or much entanglement for a given bi-partition does not make the state more or less quantum in any invariant manner. Entanglement is intrinsic to all Gaussian states of any multimode system. The difference is only that some Gaussian states contain entanglement among the most physically relevant bi-partitions, and it is common practice to reserve for them the name ``entangled states'' (see Appendix~\ref{app:qq} for more discussion).

	\subsection{Dynamics, squeezing and entanglement\label{dyn}}
	
	As mentioned before, we restrict here to Hamiltonians which are quadratic polynomials of the  canonical variables, $\hat H= \frac{1}{2} \hat r^i\,  h_{ij} \, \hat r^j+c$, where $h_{ij}$ is a symmetric, positive definite, possibly time-dependent matrix, and $c$ a constant.\footnote{Any terms linear in $\hat r^i$ can be removed by a re-definition of $\hat r^i$.} This family of Hamiltonians is the analog for a finite-dimensional system of the Hamiltonian for a free field theory discussed in the next section. These Hamiltonians preserve the Gaussianity of quantum states in the sense that the evolution of a Gaussian state with mean $\vec \mu$ and covariance matrix $\sigma$ from time $t_0$ to $t_1$,  is another Gaussian state with mean and covariance matrix  $E\cdot \vec \mu$ and  $E\cdot \sigma\cdot E^{\top}$, respectively, where $E^i_{\, j}$ is the time evolution matrix, determined from the Hamiltonian through $E={\rm T} \, \exp{\int^{t_1}_{t_0} \Omega\cdot \,h(t')\, dt'}$, with $\rm T$ indicating the standard time-ordered product. Note that  the $2N\times2N$  matrix $E$ provides the evolution of the canonical variables in the Heisenberg picture: $\hat r^i(t_1)=E^{i}_{\, j}\, \hat r^j(t_0)$. It is also exactly the same matrix that implements the Hamiltonian flow in the classical theory, in particular,  $ r^i(t)=E^{i}_{\, j}\,  r^j(t_0)$ is a solution to the classical Hamilton's equations. That the classical evolution completely determines the quantum dynamics, is a peculiarity of quadratic Hamiltonians---it is not true for more complicated Hamiltonians due to factor ordering ambiguities. Hence, evolving  Gaussian states under quadratic Hamiltonians is extremely simple: we can forget about (infinite-dimensional) density matrices,  unitary operations or Schr\"odinger's equation;  we only need to evolve its first and second moments  $(\vec \mu,\sigma)$  by multiplying them with the classical evolution matrix as indicated above. The evolution matrix $E^{i}_{\, j }$ is always an element of  the symplectic group. \\
	
	Given this background, we are now ready to study the evolution of squeezing and entanglement of a quantum state. We are interested in the analog of the question which we want to answer for scalar fields and  cosmological perturbations in sections \ref{sec:scalar} and \ref{sec:cosmpert}: if we choose a non-squeezed and non-entangled quantum state at time $t_{\rm in}$, is the time-evolved state squeezed and entangled at $t_{\rm out}>t_{\rm in}$? As emphatically discussed above, this question is unambiguous only if there are preferred quadratures at times $t_{\rm in}$ and $t_{\rm out}$. If the Hamiltonian is time-independent, the preferred set of Darboux quadrature-pairs is made of the normal modes of the system. Therefore, all we have to do to answer the question is to express the covariance matrix in this basis, and follow the time evolution of its components, $E\cdot \sigma\cdot E^{\top}$. From them, it is straightforward to compute squeezing and entanglement among different bipartitions of the preferred quadrature-pair basis.
	
	However, if the Hamiltonian does depend on time, the answer to the question requires more work, because {\em the preferred set of Darboux quadrature-pair basis may also change in time}. Therefore, from a physical standpoint, we will say that the state is not squeezed or entangled at the initial time  $t_{\rm in}$, when we find negative answers relative to the normal modes at $t_{\rm in}$. But at time $t_{\rm out}$,  what quadratures-pairs should we use to determine whether the evolved state is squeezed or entangled, the preferred quadrature-pairs at  $t_{\rm in}$ or at $t_{\rm out}$? The following simple example illustrates the difference between these two options, and helps us to understand the physical content of these two choices.\\

	{\bf Example: A time-dependent harmonic oscillator.} 
	Consider a single mode system with Hamiltonian $\hat H(t)=\frac{1}{2\, m}\, \hat p^2+\frac{1}{2} \, m\, w^2(t)\, \hat x^2$ , where $w(t)$ is time-independent in the past, then varies smoothly and monotonically, and finally becomes constant again. Let $t_{\rm in}$ ($t_{\rm out})$  be a time inside the initial (final) interval where $w(t)$ is constant, and let $w_{\rm in}$ and $w_{\rm out}$ be its initial and final values, respectively.

	At $t_{\rm in}$, the normal modes of the system are $\hat X_{\rm in}=(m\, w_{\rm in})^{1/2}\, \hat x$ and $\hat P_{\rm in}=(m\, w_{\rm in})^{-1/2}\, \hat p$. They define the annihilation  and  number operators,  $\hat A_{\rm in}=\frac{1}{\sqrt{2}}\, (\hat X_{\rm in}+i\, \hat P_{\rm in})$ and $\hat N_{\rm in}=\frac{1}{2} \, (\hat X_{\rm in}^2+\hat P_{\rm in}^2-\hat{\mathbb{I}})$, respectively. Let's assume the system starts at $t_{\rm in}$  in the ground state of the Hamiltonian $\hat H(t_{\rm in})$, which we will denote as $|{\rm in}\rangle$. This is a Gaussian state, with zero mean and covariance matrix  $\sigma_{\rm in}=\mathbb{I}_2$ equal to the identity when expressed in the Darboux basis  $\hat X_{\rm in}$,  $\hat P_{\rm in}$. Therefore, the initial state is not squeezed. We want to answer the question:  does evolution generate squeezing?
	
	To emphasize our point more clearly, let us consider two  situations: (a) We  assume that the change from $w_{\rm in}$ to $w_{\rm out}$ happens {\em adiabatically}, i.e., in a timescale $\tau$ much larger than any other natural timescale in the system, ideally  $\tau\to \infty$. (b) The change from $w_{\rm in}$ to $w_{\rm out}$ happens {\em instantaneously}. This corresponds to the limit $\tau\to 0$.
	
	In the adiabatic situation (a), the adiabatic theorem guarantees that the evolution of the state $|{\rm in}\rangle$ from $t_{\rm in}$ to $t_{\rm out}$ produces precisely the ground state of the Hamiltonian at time $t_{\rm out}$: $\hat U|{\rm in}\rangle~=~|{\rm out}\rangle$. 
	Is  $|{\rm out}\rangle$  a squeezed state? The variances of the quadrature-pair $\hat X_{\rm in}$ and $\hat P_{\rm in}$ are $ \Delta X^2_{{\rm in}} =\frac{1}{2}\,  \frac{\omega_{\rm in}}{\omega_{\rm out}}$ and $\Delta P_{{\rm in}}^2=\frac{1}{2}\,  \frac{\omega_{\rm out}}{\omega_{\rm in}}$. Therefore, if $\omega_{\rm in}\neq \omega_{\rm out}$ there is squeezing either in $\hat X_{\rm in}$ or $\hat P_{\rm in}$. Also, the expectation value of $\hat N_{\rm in}$ is different from zero at $t_{\rm out}$. 
	But these quantities do not have any natural meaning at time $t_{\rm out}$. An experimentalist entering the room at $t_{\rm out}$ will argue that the system is in the ground  state of the Hamiltonian, and that the dispersions of the normal modes at $t_{\rm out}$, namely $\hat X_{\rm out}=(m\, \omega_{\rm out})^{1/2}\, \hat x$ and $\hat P_{\rm out}=(m\, \omega_{\rm out})^{-1/2}\, \hat p$,  are both equal to $1/2$. Hence, there is neither generation of quanta nor squeezing according to this observer at time $t_{\rm out}$.
	
	In case (b), the state remains invariant, $\hat U|{\rm in}\rangle=|{\rm in}\rangle$, and therefore there is no squeezing in the initial quadrature-pair ($\hat X_{\rm in}$, $\hat P_{\rm in}$) at $t_{\rm out}$. However, the physical Hamiltonian has changed and an experimentalist at time $t_{\rm out}$  will say that the system is not in the ground state; the evolved state  is excited and squeezed with respect to the preferred quadratures at $t_{\rm out}$. \\

	These rather academic examples\footnote{
		For quantum fields during inflation, we will {\em not} assume the expansion of the universe is either adiabatic or instantaneous, but we will find that there is no generation of squeezing if one uses the quadratures singled out by the de Sitter symmetry.} reveal the importance of the choice of  quadratures in order to argue whether evolution has generated squeezing and quanta, precisely because these notions do not have an invariant meaning and are associated with observers or choices of creation and annihilation operators (a similar argument works for entanglement, although the discussion requires at least two oscillators).  
	Therefore, although mathematically one could decide to fix the quadratures once and for all and follow the evolution of their dispersions, such a strategy does not reproduce the quantities of natural interest for time-dependent systems. For them, the observables of physical interest evolve in time, and the physical characterization of a state as  squeezed or entangled must be adapted to the  evolution. This argument will be important  for fields in time-dependent spacetimes, for which additional ambiguities and mathematical subtleties arise due to the infinite number of degrees of freedom.

	\section{Dynamical generation of squeezing for a scalar field  in FLRW spacetimes \label{sec:scalar}}
	For pedagogical purposes, we consider first a real scalar field on FLRW spacetimes, since it is free of some additional complications involved in the definition of scalar and tensor curvature perturbations, which happen to be unessential for the discussion of the dynamical generation of squeezing and entanglement. In section \ref{sec:cosmpert}, we extend the discussion to include curvature perturbations.  
	The discussion of squeezing is clearer in the Schr\"odinger evolution picture \cite{Agullo:2015qqa}, and we will use it in this section---the translation to the Heisenberg picture is straightforward.\\ 
	
	

	We will work with a real scalar field, which in the classical theory satisfies  the Klein-Gordon equation
	\be 
	(\Box-m^2-\xi\, R)\, \Phi( \eta,\vec x)=0\, , \label{KG}
	\ee
	where $\Box$ is the d'Alembertian operator associated with the spacetime line element, which  in  FLRW spacetime reads $ds^2=a^2(\eta) \, (-d\eta^2+d\vec x^2)$, where $\eta$ represents conformal time, $a(\eta)$ is the scale factor, $R=6 a''/a^3$  the Ricci scalar, and $m$ and $\xi$  are real numbers representing the mass of the scalar field and its coupling to the background curvature, respectively. 
	
	The goal of this paper is to understand whether time evolution generates two-mode squeezing and entanglement between pairs of Fourier modes of the field with wavenumbers $\vec k$ and $-\vec k$. Alternatively, one can also investigate if evolution squeezes each mode individually; a discussion of  such single-mode squeezing is relegated to Appendix~\ref{app:single-mode-sqz}.  
	
	To investigate two-mode squeezing we can use the same tools as described in the previous sections for finite-dimensional systems, and apply them to the quadrature-pairs describing degrees of freedom of the fields associated with wavenumbers $\vec k$ and $-\vec k$. Hence, the  first question  we need to address is a simple one, but which contains some subtleties worth clarifying: how can we define  quadrature-pairs, or modes, associated with wavenumbers $\vec k$ and $-\vec k$ out of the field operator $\hat \Phi(\vec x)$ and its conjugate momentum $\hat \Pi(\vec x)$?  We describe two possible  strategies, of which only one turns out to be satisfactory.\\
	

	{\bf Strategy  1:} Use Fourier modes.  
	
	Let us pay attention to the Fourier components of $\hat \Phi(\vec x)$ and $\hat \Pi(\vec x)$, and define a   pair of operators labeled by $\vec k$ as 
	\be \label{phip} \hat \phi_{\vec k}=\frac{1}{(2\pi)^3} \int d^3x\, e^{-i\, \vec k\cdot \vec x}\; \hat \Phi(\vec x)\, ,  \hspace{1cm}  \hat \pi_{\vec k}= \int d^3x\, e^{-i\, \vec k\cdot \vec x}\; \hat \Pi(\vec x) \, . \ee
	
	The canonical commutation relations $[\hat \Phi(\vec x),\hat \Pi(\vec x')]=i\,  \delta(\vec x-\vec x')$ imply   $[\hat \phi_{\vec k},\hat \pi_{\vec k'}]=i\,  \delta(\vec k+\vec k')$. Hence,  for each wavenumber $ \vec k$ the operators  $\hat \phi_{\vec k}$ and  $\hat \pi_{-\vec k}$ form a canonically conjugated pair. However, these operators are not Hermitian and consequently do not describe observables. In particular we have $\hat \phi^{\dagger}_{\vec k}=\hat \phi_{-\vec k}$.  One way to bypass this impediment is by focusing on the real and imaginary parts of these operators (which are associated with the cosine and sine Fourier modes of the field): 
	\bea \label{hermitian} \hat \phi^{({\rm R})}_{\vec k}&\equiv& \frac{1}{\sqrt 2}\, (\hat \phi_{\vec k} +\hat \phi_{-\vec k} ) \, , \ \ \ \  \hat \phi^{({\rm I})}_{\vec k}\equiv  \frac{-i}{\sqrt 2}\, (\hat \phi_{\vec k} -\hat  \phi_{-\vec k} )\, , \nonumber \\ 
	\hat \pi^{({\rm R})}_{\vec k}&\equiv&  \frac{1}{\sqrt 2}\,(\hat \pi_{\vec k} +\hat  \pi_{-\vec k} )\, , \ \ \ \ \hat  \pi^{({\rm I})}_{\vec k}\equiv  \frac{-i}{\sqrt 2}\, (\hat \pi_{\vec k} -\hat  \pi_{-\vec k} )\, . \eea
	%
	These operators are Hermitian, and the only non-vanishing commutation relations between them are
	\be [\hat \phi^{({\rm R})}_{\vec k},\hat \pi^{({\rm R})}_{\vec k'}]=i\, \delta(\vec k+\vec k') \, , \ \ \ [\hat \phi^{({\rm I})}_{\vec k},\hat \pi^{({\rm I})}_{\vec k'}]=i\, \delta(\vec k+\vec k') \, , \ee
	and, consequently, for each $\vec k$ they define two canonical quadrature-pairs. However, because the  operators $\hat \phi^{({\rm R})}_{\vec k}$ and $\hat \pi^{({\rm R})}_{\vec k} $ are invariant under $\vec k \to - \vec k$, while  $\hat \phi^{({\rm I})}_{\vec k}$ and $\hat \pi^{({\rm I})}_{\vec k} $ change sign (as expected from cosine and sine modes), when working with these Hermitian fields one must restrict to half of the wavenumber space,
	\footnote{\label{Z+}More concretely, the Hermitian quadratures are defined for wavenumbers  $\vec k \in  k_0 \times \, \mathbb{R}^3_{(+)}$, where $k_0$ is a real number with dimensions of inverse length and 
		$\mathbb{R}^3_{(+)}\equiv\{(k_x,k_y,k_z)\in \mathbb{R}^3 \,  :\,  k_z>0\} \cup \{(k_x,k_y,k_z)\in \mathbb{R}^3  \,  :\,   k_z=0, k_y >0\} \cup \{(k_x,k_y,k_z)\in \mathbb{R}^3  \,  :\,  k_z=0, k_y =0, k_x>0\}$ \cite{Ashtekar:2009mb}. } since the other half does not describe independent degrees of freedom. (Of course the number of degrees of freedom remains the same, as simple counting reveals.) Therefore, for these Hermitian pairs of operators it does not make sense to talk about two-mode squeezing or entanglement between modes $\vec k$ and $-\vec k$.  One could instead discuss two-mode squeezing and entanglement between ($\hat \phi^{({\rm R})}_{\vec k},\hat \pi^{({\rm R})}_{\vec k} $) and ($\hat \phi^{({\rm I})}_{\vec k},\hat \pi^{({\rm I})}_{\vec k} $). But a simple calculation shows that for  Gaussian states  all  cross-correlations between these two pairs vanish at any time, due to the ``orthogonality'' of the cosine and sine modes. 
	Hence, there is no squeezing or entanglement generated between them in FLRW spacetimes, no matter what the expansion of the universe is. We therefore conclude that the Fourier components of the field are not suitable  variables  for our purposes.\\ 
	
	{\bf Strategy  2:} Use creation and annihilation operators and build Hermitian quadrature-pairs from them via 
	\be \hat X_{\vec k}\equiv\frac{1}{\sqrt{2}} \, (\hat A_{\vec k}+\hat A_{\vec k}^{\dagger})\, , \ \ \ \hat P_{\vec k}\equiv- \, \frac{i}{\sqrt{2}} \, (\hat A_{\vec k}-\hat A_{\vec k}^{\dagger})\, . \ee  (The operator  $\hat X_{\vec k}$ should not be confused with $\hat \phi_{\vec k}$, which is not Hermitian, and when expanded in terms of annihilation and creation takes a different form, namely $\hat{\phi}_{\vec{k}} = f_{k} \, \hat{A}_{\vec{k}} + \bar{f}_k \, \hat{A}^{\dagger}_{-\vec{k}}$, for appropriately normalized mode functions $f_k$. Even the dimensions of $ \hat X_{\vec k}$ and $\hat \phi_{\vec k}$ are different.)
	
	These Hermitian operators  are canonically conjugate, $[\hat X_{\vec k}, \hat P_{\vec k'}]=i\, \delta(\vec k+\vec k') $. They are defined for all $\vec k$, and $\hat X_{\vec k}$ is independent of $\hat X_{-\vec k}$ ---and not related to it by Hermitian conjugation (for the same reason that  $\hat A_{\vec k}$ is independent of $\hat A_{-\vec k}$).  Consequently, these operators allow us to define bi-partitions, squeezing and entanglement between modes $\vec k$ and $-\vec k$ in a mathematically well-defined manner. This is indeed  the strategy used in some of the previous literature (see e.g.\ 
	\cite{Martin:2015qta}). The drawback is that there is huge ambiguity in the definition of $\hat X_{\vec k}$ and $\hat P_{\vec k}$, precisely the well-known ambiguity associated with the definition of $\hat A_{\vec k}$ and $\hat A^{\dagger}_{\vec k}$ or, equivalently, in the notion of vacuum and particles in FLRW spacetimes.\footnote{For finite-dimensional systems, one can always define preferred quadratures at any chosen time $\eta_0$, even if the Hamiltonian is time dependent,  by using the instantaneous normal modes of the Hamiltonian. This defines annihilation operators $\hat A_{\vec k}$ at $\eta_0$,  whose associated vacuum is the instantaneous ground state of the Hamiltonian. This instantaneous diagonalization of the Hamiltonian, although it is a licit strategy for finite-dimensional systems, presents numerous problems in field theory, even in the simple case of FLRW spacetimes \cite{Fulling:1979ac}: in addition to the large ambiguity in the definition of a canonical Hamiltonian of a field theory in a time-dependent spacetime, the strategy fails to produce finite particle creation in the course of time and renormalizability of the energy-momentum tensor, for generic forms of the scale factor $a(\eta)$ and arbitrary values of $m$ and $\xi$.  These arguments are well-known \cite{Fulling:1979ac}, although often overlooked.} 
	One of the key goals of this paper is to emphasize this ambiguity in the definition of quadrature-pairs $\hat X_{\vec k}, \hat P_{\vec k}$ associated to each Fourier mode $\vec k$, and to argue that it translates to an ambiguity in the definition of squeezing and entanglement between the sectors $\vec k$ and $-\vec k$ of the field. 
	
	To better understand the impact of this ambiguity, it is illustrative to write explicitly how it affects the quadrature-pairs  $\hat X_{\vec k}, \hat P_{\vec k}$. The ambiguity in the definition of  $\hat A_{\vec k}$ and $\hat A^{\dagger}_{\vec k}$  reduces to 
	\be \label{bog}  \hat{ A}'_{\vec k}=\alpha_k\, \hat A_{\vec k}+\beta_k\, \hat A^{\dagger}_{-\vec k}\, , \ee
	where $\alpha_k$ and $\beta_k$ are complex numbers which depend only on the modulus of $\vec k$ and satisfy the normalization condition $|\alpha_k|^2-|\beta_k|^2=1$. Any set of such operators $\{\hat{ A}'_{\vec k}\}_{\vec k}$ defines a legitimate Fock vacuum that is invariant under translations and rotations, and no choice is preferred, except when the expansion of the universe is very special. The important aspect of this ambiguity is that {\em it mixes the $\vec k$ and $-\vec k$ sectors}. This can be seen  explicitly by writing the relation between the quadrature-pairs defined from $\hat A_{\vec k}$ and $\hat A'_{\vec k}$:
	\bea \label{mix} \hat{ X}'_{\vec k}&=&{\rm Re}[\alpha_{k}] \,\hat X_{\vec k}-{\rm Im}[{\alpha_{k}}] \,\hat  P_{ \vec k}+{\rm Re}[{ \beta_{k}}] \, \hat X_{ -\vec k} +{\rm Im}[{\beta_{k}}] \, \hat P_{ -\vec k}\, , \nonumber \\
	\hat {P}'_{ \vec k}&=&{\rm Re}[{ \alpha_{k}}] \, \hat P_{ \vec k}+{\rm Im}[{\alpha_{k}}] \, \hat X_{\vec k}-{\rm Re}[{\beta_{k}}] \, \hat P_{ -\vec k}+{\rm Im}[{ \beta_{k}}] \, \hat X_{-\vec k}\, .\eea
	%
	This expression shows that the ``primed'' quadrature-pair for $\vec k$ is a mix of the original quadrature-pairs for $\vec k$ {\em and} $-\vec k$. This implies that we can obtain any answer we want for the degree of two-mode squeezing and entanglement between $\vec k$ and $-\vec k$ by appropriately choosing the quadrature-pairs  $(\hat X'_{\vec k}, \hat P'_{\vec k})$. In particular, there always exists a choice for which any homogeneous and isotropic Gaussian state is manifestly unsqueezed and unentangled. 
	
	One could still think that, although it is true that there exists ambiguity at a given time, one can still unambiguously talk about the {\em generation} of squeezing and entanglement during evolution as follows: fix a choice of  quadrature-pairs once and for all, and compare their properties in the state before and after the evolution. If initially the state is unsqueezed and unentangled between the $\vec k$ and $-\vec k$  sectors and after the evolution this is no longer true, one can say that the evolution has squeezed and entangled these two sectors. This strategy is problematic in time-dependent spacetimes, since generically it involves the use of non-Hadamard states, as we argue below. Moreover, this strategy ignores another ambiguity: there is no reason to choose the same quadrature-pairs at the initial and final times, and in general one should not. This may seem counter-intuitive and unnatural at first, and we now provide two reasons to argue why this is indeed the case, one physical and one mathematical. 
	
	(i) From the physical viewpoint, because the properties of the spacetime change due to the expansion, the most natural choice of creation and annihilation variables ($\hat{A}_{\vec k},\hat{A}^{\dagger}_{\vec k}$)  and quadrature-pairs constructed from them also changes. To evaluate whether the final state is squeezed and entangled, it is natural to use the preferred notions at the final time (if such exist); this is analogous  to the example of the time-dependent oscillator discussed in section~\ref{dyn}. In fact, this is what is done in the most well-known examples of particle creation in curved spacetimes, namely the Hawking effect \cite{hawking_particle_1975} and Parker's particle creation in  FLRW spacetimes that are asymptotically Minkowskian in the past and future \cite{Parker:1968mv,Parker:1969au,Parker:1971pt}. In the Hawking effect, one chooses an initial state that is the vacuum relative to the preferred creation and annihilation variables in the asymptotic past, but probes the properties of the final state after the evolution using the natural choice  of creation and annihilation variables in the {\em future} asymptotic region. This is indeed the natural strategy from a physical perspective. In such a scenario, one finds that the final state contains particles, and the degrees of freedom escaping the black hole  are squeezed and entangled  with those falling into the horizon. Similarly, in Parker's asymptotically Minkowskian FLRW, one has preferred quadrature-pairs in the past and future, but they are different. One uses the preferred choice in the past to prepare the state and the preferred choice in the future to probe it after the evolution. One also finds that  there is particle creation for general expansion histories, which is accompanied by squeezing and entanglement between the $\vec k$ and $-\vec k$ sectors. 
	
	Hawking's and Parker's examples have permeated the intuition of many physicists, due to its simple interpretation in terms of particles. But this interpretation rests crucially on the assumption of asymptotically Minkowskian regions, which allow us to select  preferred choices of creation and annihilation variables in the past and future, respectively. In our universe, such regions are not available and one must face  the ambiguity in the choice of creation and annihilation variables. The examples of Hawking radiation and Parker's particle creation teach us that there is no reason to use the same notion of particles and choice of quadrature-pairs at initial and final times.
	
	(ii) Recall that we work in the Schr\"odinger picture. A choice of creation and annihilation variables at a time $t$ defines a vacuum state at that time, and the associated notion of particles and quadrature-pairs. In making such a choice, there are certain restrictions one must follow. In arbitrary spacetimes, it is accepted that  any permissible vacuum must be a Hadamard state \cite{Wald:1995yp}. This guarantees that the short distance behavior of the state has the appropriate physical and mathematical form, which in turn allows one to recover  results compatible with local Lorentz invariance at short distances, and  to renormalize the ultraviolet divergences that appear in calculations of the energy-momentum tensor and other composite operators. The Hadamard condition depends on the characteristics of the geometry at the time it is applied. Consequently, a Schr\"odinger state that is Hadamard at $t_1$ is in general not Hadamard at time $t_2$ if the spacetime is time dependent. Therefore,  if one uses the same choice for $\hat{A}_{\vec k}$ and $\hat{A}^{\dagger}_{\vec k}$, and the same notion of quadrature-pairs at all times, one is involving mathematical structures that violate the Hadamard condition. This is accompanied by some well-known issues; particularly, one would generically find that infinitely many particles are created per unit volume of space and that  the energy-momentum tensor is not renormalizable.  One could ignore this issue by arguing that these are ``ultraviolet problems'', which can be hidden by introducing a cut-off. The introduction of a cut-off, however, introduces other problems (breakdown of unitarity, no conservation of  energy and momentum, etc.). 
	These mathematical complications are peculiar to field theory in time-dependent spacetimes. They neither arise for fields in Minkowski spacetime nor for time-dependent finite-dimensional systems, and it is for this reason that they are sometimes overlooked.

	We reach the conclusion that in general FLRW spacetimes the question {\em``does time evolution produce two-mode squeezing and entanglement between  the  $\vec k$ and $-\vec k$ degrees of freedom?''}, is equivalent to the question {\em``does the evolution create particles?'', and they both suffer from the same ambiguity, namely the definition of particles at initial and final times.}

	Nevertheless, in special situations for which additional symmetries exist, one can take advantage of them to find a preferred choice. We will now discuss an example of direct relevance for this paper: a scalar field in the Poincar\'e patch of de Sitter spacetime. \\
	
	\subsection{Example: A  scalar field in the Poincar\'e patch of de Sitter spacetime\label{subsec:PdS}}
	Consider a spatially flat FLRW universe  with a scale factor of the form $a(\eta)=-\frac{1}{H\, \eta}$ with $\eta$ denoting the conformal time and $H$ a constant (in proper time, $a(t)=a_0\, e^{H\, t}$). We will refer to it as the Poincar\'e patch of de Sitter spacetime (PdS). This spacetime admits, in addition to the six isometries common to all FLRW geometries accounting  for homogeneity and isotropy, one extra Killing vector field associated with the de Sitter group.\footnote{The de Sitter group has ten independent Killing vectors fields, and all of them, locally, are isometries of PdS. However, the PdS  is only a portion of de Sitter spacetime, so not all ten transformations are global isometries of PdS. Only the subgroup of the de Sitter group which leaves the Poincar\'e patch invariant describes the global isometries of  PdS.  See  \cite[Sec.~IV~C]{abk1} and Appendix A for further details.}  
	Although this extra isometry is not globally time-like, it suffices to single out a preferred notion of vacuum at a given time (when complemented with the Hadamard condition); this is the so-called Bunch-Davies (BD) vacuum \cite{Chernikov:1968zm,Tagirov:1972vv,Bunch:1978yq} (see Appendix \ref{app:BD} for  details omitted in this section). In the usual terminology, the BD vacuum at time $\eta_0$ is defined from the initial data $(e^{\rm BD}_k(\eta_0), {\partial_{\eta}e_k^{\rm BD}}(\eta_0))$ of the Bunch-Davies solutions to the equations of motion, or mode functions:\footnote{More precise, the BD vacuum at time $\eta_0$  is the normalized state annihilated by the operators 
		\be \hat A_{\vec k}=i\, \big(e^{\rm BD\, *}_k(\eta_0)\, \hat \pi_{\vec k}-(2\pi)^3\, a(\eta_0)^2\, {\partial_{\eta}e^{\rm BD\, *}_k}(\eta_0)\, \hat \phi_{\vec k}\big)\, , \ee
		for all values of  $\vec k$.} 
	\be \label{BDmodes} e^{\rm BD}_k(\eta)= 
	\sqrt{\frac{-\pi \, \eta}{4 \, (2\pi)^3\, a(\eta)^2}}\, H^{(1)}_{\mu}(-k\eta)\, , \ee
	%
	where $H^{(1)}_{\mu}(x)$ is a Hankel function with index $\mu^2=\frac{9}{4}-\frac{m^2}{H^2}-12\xi$ (recall the BD state is ill-defined for $m=0$ and $\xi=0$ \cite{Allen:1985ux}). One important point to notice is that, in the Schr\"odinger picture, there is not a single BD vacuum, but rather {\em a BD vacuum at each instant of time}, which we will denote by $|{\rm BD},\eta\rangle$. The BD vacuum at $\eta_0$,  $|{\rm BD},\eta_0\rangle$, is defined from the  initial data $(e^{{\rm BD}}_k(\eta_0), {\partial_{\eta} e_k^{{\rm BD}}}(\eta_0))$, while $|{\rm BD},\eta_1\rangle$ is defined from  $(e^{{\rm BD}}_k(\eta_1), {\partial_{\eta} e_k^{{\rm BD}}}(\eta_1))$. Since these two sets of  initial data are different, these  are different states in the Schr\"odinger picture. 
	The state $|{\rm BD},\eta_0\rangle$ is invariant under the PdS isometries and Hadamard only at $\eta_0$, and it is the only state with such properties at $\eta_0$. Moreover, the  one-parameter family of states  $|{\rm BD},\eta\rangle$ are connected by time evolution: $|{\rm BD},\eta_1\rangle=\hat U_{\eta_1\eta_0}|{\rm BD},\eta_0\rangle$ (these statements are proven in Appendix \ref{app:BD}). These points go  unnoticed if one works in the Heisenberg picture, which is far more common in textbooks, since there one simply fixes the state  $|{\rm BD},\eta_0\rangle$ at $\eta_0$ once and for all,  and refers to it as {\em the} BD vacuum.  
	
	Therefore, the symmetries of PdS, when complemented with the Hadamard condition, provide a preferred choice of  quadratures {\em at each instant of time} $\eta$:
	\be \label{BDq} \hat X_{\vec k}^{(\eta)}=\frac{1}{\sqrt{2}} \, \Big(\hat A^{(\eta)}_{\vec k}+\hat A_{\vec k}^{(\eta)\, \dagger}\Big), \qquad \hat P^{(\eta)}_{\vec k}=-\, \frac{i}{\sqrt{2}} \, \Big(\hat A^{(\eta)}_{\vec k}-\hat A_{\vec k}^{(\eta)\, \dagger}\Big)\, , \ee
	%
	%
	where $\hat A_{\vec k}^{(\eta)}$ is defined from  $(e^{{\rm BD}}_k(\eta), \partial_{\eta}{e_k^{{\rm BD}}}(\eta))$, and hence annihilates the state $|{\rm BD},\eta\rangle$  for all wavenumbers $\vec k$. We emphasize that the label $\eta$ in these operators should not be interpreted as  Heisenberg evolution, since it is not---for the same reason that the quadrature $\hat X_{\rm out}$ in the example of the time-dependent harmonic oscillator in section \ref{dyn} is not the Heisenberg evolution of  $\hat X_{\rm in}$ and that the $out$ number operator $\hat N_{out}$ in the Hawking effect is not the time evolution of the $in$ number operator $\hat N_{in}$.

	Therefore, if one decides to use the symmetries of the PdS to resolve the ambiguity in the discussion of squeezing and entanglement, one must use the quadrature $\hat X_{\vec k}^{(\eta_0)},\hat P_{\vec k}^{(\eta_0)}$ at $\eta_0$, and the quadratures $\hat X_{\vec k}^{(\eta_1)},\hat P_{\vec k}^{(\eta_1)}$ at $\eta_1$. But as mentioned before, if the system is prepared at time $\eta_0$ in the state $|{\rm BD},\eta_0\rangle$, the evolution brings it to $|{\rm BD},\eta_1\rangle$. This  automatically implies that the final state contains no particles and there is no squeezing or entanglement between the sectors $\vec k$ and $-\vec k$ as defined  by the preferred quadratures at the final time. Thus, although we have used exactly the same strategy one uses in the Hawking effect, in PdS there is no particle creation. The difference is of course the high degree of symmetry  of PdS. 
	
	Based on this result, one should not conclude that the generation of the primordial density perturbations during inflation does not generate any feature genuinely quantum, such as entanglement. It simply tells us that there are no particles created for the notion of particles singled out by the symmetries of PdS, and correspondingly there is no entanglement given this notion of particles. Nevertheless, this does not necessarily imply that there is no entanglement between the degrees of freedom that we have access to in observations, namely the field in real space. We further elaborate on this in section \ref{sec:real}.
	

	\subsection{Connection with previous work}\label{sec:previous-work}
	
	The conclusions reached  in the example  above  contrast with previous discussions (see e.g. \cite{Albrecht:1992kf,Lesgourgues:1996jc,Martin:2015qta,Hsiang:2021kgh}). We explain here the origin of the differences. In short, in previous references there is an implicit choice of quadrature-pairs at early times, and another choice at late times. We argue in the following  that these choices are not natural in any sense, that they are not compatible with the (approximated) symmetries of slow-roll inflation, and that they also have some undesirable mathematical features. \\
	
	In discussing the generation of squeezing during inflation, it has been common to define quadratures-pairs using the following argument---although most times only in an implicit way. Recall first that  all spatially-flat FLRW metrics are conformally related to the Minkowski line element. This in turn implies the following (see e.g. \cite{waldGR}): if a scalar field $\Phi(\eta,\vec x)$ satisfies 
	\be \big[\Box-m^2-\xi\, R\big]\,  \Phi( \eta,\vec x)=0\, , \ \ {\rm with}\ \  \Box=g^{\mu\nu}\nabla_{\mu}\nabla_{\nu}\, , \ee then, given any nowhere vanishing  smooth function $\Omega(\eta,\vec x)$, the re-scaled field $ v(\eta,\vec x)=\Omega^{-1}\,  \Phi(\eta,\vec x)$ satisfies the equation
	\be \big[\tilde{\Box}- \Omega^{-2}\, m^2-(\xi-\frac{1}{6})\, \Omega^{-2}\, R\big]\,  v(\eta,\vec x)=0\, , \ \  {\rm where}\ \  \tilde{\Box}\equiv \tilde g^{\mu\nu}\tilde{\nabla}_{\mu}\tilde{\nabla}_{\nu}\, , \ee and  $ \tilde g_{\mu \nu}=\Omega^2\, g_{\mu\nu}$ is the conformally-rescaled metric. Hence, if we choose $\Omega=a(\eta)^{-1}$, the field $v(\eta,\vec x)\equiv a(\eta)\,  \Phi(\eta,\vec x)$ satisfies the Klein-Gordon equation in Minkowski spacetime, with a time-dependent potential: $(\partial^{\mu}\partial_{\mu} - V(\eta) )\, \hat v(\eta,\vec x)=0$, where the explicit form of the potential  is $V(\eta)=a^2(\eta)\, m^2+(\xi-\frac{1}{6}) \, 6\, \frac{a''}{a}$ (we have used that $R=6 \frac{a''}{a^3}$). In this way, the time dependence of the FLRW line element can be traded off by a time-dependent potential.\footnote{The introduction of the variable $v(\eta)$ is usually motivated by saying that  its equation of motion in Fourier space, when expressed in conformal time, does not contain terms proportional to the first time derivative of the field: $v_k''(\eta)+f(k,\eta) \, v_k(\eta)=0$, thereby simplifying to the equation of motion of a time-dependent harmonic oscillator. Note, however, that  $v(\eta)$ is not the only variable with this property (e.g.\ it is also true for $\chi(\vec x,\tau) \equiv a^3(\tau)\, \Phi(\vec x,\tau)$ when working with harmonic time $\tau$, defined as $dt=a^3\, d\tau$). What makes $v(\eta,\vec x)$ unique is the fact that it satisfies, in position space, the Klein-Gordon equation with respect to the Minkowski spacetime metric with a time-dependent potential.} One could then forget about the potential and define quadratures using the isometries of the Minkowski metric. This is done by defining annihilation  operators $\hat A^M_{\vec k}$ using initial data 
	\be \label{M} v_k(\eta_0)=\frac{1}{(2\pi)^{3/2}}\,\frac{1}{\sqrt{2 k}}\, , \ \ \ \ v'_k(\eta_0) =\frac{1}{(2\pi)^{3/2}}\,\frac{-i\, k }{\sqrt{2 k}}\,  \ee
	and defining quadrature-pairs from them:  $\hat x^{M}_{\vec k}=\frac{1}{\sqrt{2}} \, (\hat A^{M}_{\vec k}+\hat A_{\vec k}^{M\, \dagger})$, $\hat p^{M}_{\vec k}=-\, \frac{i}{\sqrt{2}} \, (\hat A^{M}_{\vec k}-\hat A_{\vec k}^{M\, \dagger})$. 
	The analysis of \cite{Albrecht:1992kf,Lesgourgues:1996jc,Martin:2015qta,Hsiang:2021kgh} shows that, if the system is prepared in the state annihilated by $\hat A^{M}_{\vec k}$ at some time $\eta_0$ during inflation ---hence the state has no squeezing or entanglement between the pairs  $\hat x^{M}_{\vec k},\hat p^{M}_{\vec k}$ and $\hat x^{M}_{-\vec k},\hat p^{M}_{-\vec k}$ at $\eta_0$--- the  inflationary evolution acts like a two-mode squeezer for these two quadrature-pairs, generating squeezing and entanglement. Although this is mathematically true, we make the following observations:
	
	(i) The physical meaning of  the variable $v(\eta,\vec x)$ is obscure,  since it is constructed by multiplying the scalar field $\Phi$ by the scale factor $a(\eta)$. Since the value of $a(\eta)$ can be re-scaled arbitrarily by a mere change of coordinates,  so can $v(\eta,\vec x)$. In particular, $v(\eta,\vec x)$ does not transform like a scalar field, or any other covariant quantity under diffeomorphisms. The variable $v(\eta,\vec x)$ has a clean physical meaning only in the special situation of a conformally coupled massless scalar field $(m=0,\, \xi=\frac{1}{6})$.
	
	(ii) The choice of quadrature-pairs based on \eqref{M} is motivated by  the symmetries of the auxiliary Minkowski metric and neglects the time-dependent potential $V(\eta)$. This amounts to 
	ignoring the actual time-dependence of the physical spacetime. 
	
	(iii) If this initial choice of quadrature-pairs (based on \eqref{M}) is translated to the physical field $\Phi$, it actually corresponds to a {\em time-dependent choice}  $\hat X^{(M,\eta)}_{\vec k}=a(\eta) \, \hat x^{M}_{\vec k}$ and $\hat P^{(M,\eta)}_{\vec k}=\frac{1}{a(\eta)}\, \hat p^{M}_{\vec k}$, because the relation between $\hat v$ and $\Phi$ is time-dependent. Thus, although one may have the impression of working with a fixed choice of quadrature-pairs when using \eqref{M}, from the viewpoint of the physical field $\Phi$ one is actually working with time-dependent quadrature-pairs.
	
	(iv) The vacuum state selected by the quadratures  \eqref{M}  is not Poincar\'e-de Sitter invariant.
	
	(iv) The choice \eqref{M} is not only not preferred, but it contains undesirable features: the vacuum state it selects is not a Hadamard state \cite{Wald:1995yp}. More specifically, it is a state of zeroth adiabatic order, as defined by Parker and Fulling \cite{Parker:1974qw,Parker:2009uva}. Consequently, for this state the energy-momentum tensor is non-renormalizable using local and covariant methods \cite{Wald:1995yp}. \\

	In spite of these features, one could argue that there is nothing actually incorrect in using the quadrature-pairs $\hat X^{(M,\eta)}_{\vec k}, \hat P^{(M,\eta)}_{\vec k}$ for the task of  quantifying the entanglement generated during inflation between $\vec k$ and $-\vec k$ modes, and we do not completely disagree with that. In other words, our goal here is not to identify any mistake in the previous literature, but rather to emphasize the huge ambiguity in speaking about squeezing and entanglement in Fourier space. Because of this ambiguity,  the answer one obtains tells us more about the concrete choice made to define quadrature-pairs than about the invariant and physical properties of the quantum state. Our conclusion is therefore that it would be more fruitful to leave aside Fourier space when speaking about generation of entanglement during inflation and rather focus attention on the degrees of freedom which we can actually observe, namely the field in {\em real space}. Centering the discussion on concrete physical observables would help to remove the ambiguities, and to quantify in an invariant manner the entanglement that could have accompanied the generation of the primordial density perturbations during inflation. We discuss this in section \ref{sec:real}.

	\section{Cosmological perturbations and inflation\label{sec:cosmpert}}
	
	The goal of this section is to extend the previous discussions to cosmological perturbations during inflation. We will focus on scalar perturbations for brevity, but all conclusions apply equally to tensor perturbations.  When compared to the scalar field in PdS spacetime, two additional subtleties appear: (i) the issue of gauge freedom in the definition of scalar perturbations, and (ii) the fact that an inflationary spacetime deviates from an exact PdS spacetime. The goal of this section is to check that these two subtleties do not modify the conclusions reached in previous sections. 
	
	Consider an inflationary spacetime, where the matter content is given by a scalar field  $\varphi$, the inflaton, subject to a potential $V(\varphi)$ compatible with slow-roll inflation. Let $\epsilon=-\frac{\dot H}{H^2}$ and $\delta=\frac{\ddot H}{2\dot H H}$ be the slow-roll parameters, which are assumed to be $\epsilon,\delta \ll 1$ during inflation. In particular, this implies that their time derivatives can be neglected, so they will be treated as constants; this is the so-called slow-roll approximation. This facilitates finding analytic expressions for the vacuum state.

	The most commonly used variable to describe scalar perturbations during inflation is the comoving curvature perturbation field  $\mathcal{R}(\eta,\vec x)$. It is related to the perturbations of the inflaton field $\delta \varphi$ and the Bardeen potential $\Psi$ through $\mathcal{R}=\Psi+\frac{H}{\dot \varphi}\, \varphi$ (see, e.g. \cite{Weinberg:2008zzc}). The equation of motion for $\mathcal{R}$  is obtained by expanding Einstein's equations to linear order in the perturbation; in Fourier space, it reads
	\be \label{eqR}  \mathcal{R}''_{\vec k}(\eta)+2\frac{z'}{z}\,  \mathcal{R}'_{\vec k}(\eta)+k^2\, \mathcal{R}_{\vec k}(\eta)=0\, ,\ee
	where $z(\eta)\equiv a\, \frac{\dot \varphi}{H}$. 
	The main advantages of the variable $\mathcal{R}$ are: (i) It is gauge invariant at first order in perturbations. (ii) It has a direct physical meaning: it describes the curvature of the $\eta=$ constant spatial sections, $R^{(3)}(\vec x)=\frac{4}{a^2}\, \nabla^2\mathcal{R}(\vec x)$, or in Fourier space $R^{(3)}_{\vec k}=-4\frac{k^2}{a^2}\, \mathcal{R}_{\vec k}$. (iii) It is time-independent for super-Hubble modes \cite{Weinberg:2008zzc}---this allows us to identify the value of the perturbations at the end of inflation with its value at horizon re-entry during the radiation or matter dominated eras. 
	
	The main difference with the previous section is that the exact de PdS invariance is lost, because $\epsilon,\delta \neq 0$, i.e.\ $H$ is no longer constant. Therefore, strictly speaking, the extra symmetry  which one uses to single out a preferred vacuum and quadrature-pairs for $\vec k$ and $-\vec k$, is not available. However, since $\epsilon,\delta \ll 1$, the common strategy is to use the approximate PdS invariance to extend the arguments of the previous section by replacing the Bunch-Davies modes of Eq.~\eqref{BDmodes} by their slow-roll generalization:\footnote{Recall  that comoving curvature perturbations cannot be defined for exact PdS spacetime, i.e., in the limit $\epsilon\to 0$. Also, even for non-zero $\epsilon$,  the state must be modified for $|\vec k|\to 0$ to avoid infrared divergences (see e.g. \cite{Agullo:2009zi}). This implies that, strictly speaking, there is no PdS invariant state for cosmological perturbations. However, since the deviations from PdS invariance occur for $|\vec k|\to 0$, they are not accessible with observations within our finite patch of the universe, and because of this it is normally argued that the state is PdS invariant for ``all practical purposes''. }
	\be \label{BDinlf} e^{{\rm BD}}_k(\eta)=\sqrt{\frac{-\pi \, \eta}{4 \, (2\pi)^3 \, z(\eta)^2}}\, H^{(1)}_{\mu}(-k\eta)\, , \ee
	where $\mu=\frac{3}{2}+2\epsilon +\delta$.  It is common to keep using the name  Bunch-Davies for these solutions and the vacua they define, and use them to define preferred quadratures associated with $\vec k$ and $-\vec k$ at each instant of time, as we did in the previous section. More explicitly, in the Schr\"odinger picture, (\ref{BDinlf}) defines a one-parameter family of vacua, which we will denote as $|{\rm BD},\eta\rangle$. One also has the associated family of annihilation operators $\hat A^{(\eta)}_{\vec k}$ and  quadrature-pairs $\hat X^{(\eta)}_{\vec k}=\frac{1}{\sqrt{2}}(\hat A^{(\eta)}_{\vec k}+\hat A^{(\eta)\, \dagger}_{\vec k})$ and  $\hat P^{(\eta)}_{\vec k}=\frac{-i}{\sqrt{2}}(\hat A^{(\eta)}_{\vec k}-\hat A^{(\eta)\, \dagger}_{\vec k})$, in terms of which the number operators read $\hat N^{(\eta)}_{\vec k}=\frac{1}{2} \big[( \hat X^{(\eta)}_{\vec k})^2+(\hat P^{(\eta)}_{\vec k})^2-\hat{\mathbb{I}}\big]$.  The  states in the family $|{\rm BD},\eta\rangle$ are related by  time evolution in Schr\"odinger's picture, $|{\rm BD},\eta_1\rangle=\hat U_{\eta_1,\eta_0}|{\rm BD},\eta_0\rangle$. Then, if perturbations are in the state $|{\rm BD},\eta_0\rangle$ at time $\eta_0$, for which $\langle \hat N^{(\eta_0)}_{\vec k}\rangle=0$ for all $\vec k$, at time $\eta_1>\eta_0$ the state evolves to $|{\rm BD},\eta_1\rangle$, for which $\langle \hat N^{(\eta_1)}_{\vec k}\rangle=0$ for all $\vec k$, and there is neither entanglement nor squeezing between the pairs $(\hat X^{(\eta_1)}_{\vec k},\hat P^{(\eta_1)}_{\vec k})$ and $(\hat X^{(\eta_1)}_{-\vec k},\hat P^{(\eta_1)}_{-\vec k})$. \\
	
	In summary, whether inflation squeezes and entangles scalar curvatures perturbations is an ambiguous question, and the answer depends on the choice of quadratures associated with the wavenumbers  $\vec k$ and $-\vec k$ at the initial and final times. If one takes advantage of the approximate de Sitter symmetries to define preferred vacua, particles and quadratures at each instant, there is neither particle creation nor generation of squeezing and entanglement. This result, however, does not make the state of perturbations at the end of inflation  less quantum in any invariant manner.

	\section{Correlations and entanglement in real space\label{sec:real}}
	
	Fourier space is a useful tool to compute many aspects of field theory in FLRW spacetimes. However, for the generation of squeezing and entanglement during inflation it is crucial to  pay attention to the observables we have access to because, as emphasized above,  squeezing and entanglement  are not properties of the quantum state alone.  
	We observe in real space, and whether there is entanglement between modes $\vec k$ and $-\vec k $ in Fourier space is not of direct physical relevance. Therefore, we are interested in whether inflation creates entanglement between the degrees of freedom of the primordial curvature perturbations associated with  different regions of space. Even in real space we need to formulate the question with some care, because: (i) Entanglement between spatially separated regions is ubiquitous in quantum field theory, even for the vacuum in Minkowski spacetime \cite{reeh1961bemerkungen,2018}. One possible avenue to isolate what inflation is adding, is to compare the entanglement at  the end of inflation with what one would find in Minkowski spacetime for the ``same two regions''. (ii) To discuss entanglement, we first need to define the two subsystems we are interested in. A field has infinitely many degrees of freedom and  we can only access a few in observations, out of which we want to define our two subsystems. Recall that a physical subsystem can be identified with a set of {\em pairs} of canonically conjugated observables---more precisely, with the subalgebra they define~\cite{haag_local_1996}. In the problem we are considering, we can obtain such pairs by ``averaging'' (or smearing) the field  and its conjugate momenta in a region of space $\mathfrak{R}$
	\be
	\hat  \Phi_\mathfrak{R} \equiv \int_\mathfrak{R} d^3x\, f(\vec x)\, \hat \Phi(\vec x)\, , \qquad 
	\hat \Pi_\mathfrak{R} \equiv \int_\mathfrak{R} d^3x\, g(\vec x)\, \hat \Pi(\vec x)\, , 
	\ee
	where $f(\vec x)$ and $g(\vec x)$ are two functions of compact support restricted  to the region of space $\mathfrak{R}$. This region can be thought of as the minimum resolution of our detectors, and the functions $f(\vec x)$ and $g(\vec x)$ are determined by the properties of the detector. The commutation relations of these two observables are 
	\be 
	[\hat \Phi_\mathfrak{R},\hat \Pi_\mathfrak{R}]= i \int d^3x \, f(\vec x)\,  g(\vec x) \, ,  
	\ee
	i.e., the overlap of the two smearing functions $f(\vec x)$ and $g(\vec x)$. Hence, if these functions are such that  $\int d^3x f(\vec x)\,  g(\vec x)=1$,  the operators $(\hat \Phi_\mathfrak{R},\hat \Pi_\mathfrak{R})$ form a canonical pair, and define a ``single-mode'' subsystem (classically this subsystem corresponds to a two-dimensional subspace of the phase space).\footnote{It is important to notice that we can extract (infinitely) many different canonical pairs from any finite region of space $\mathfrak{R}$. This can be done by choosing smearing functions $f'(\vec x)$ and $g'(\vec x)$ with compact support within $\mathfrak{R}$ and whose integrals against  $f(\vec x)$ and $g(\vec x)$ vanish. There are infinitely many such choices. This is in agreement with the well-known fact that any region of space $\mathfrak{R}$ hosts infinitely many degrees of freedom of the field. Each smearing extracts a single one of them.} Given two such pairs commuting among themselves, $(\hat \Phi_\mathfrak{R},\hat \Pi_\mathfrak{R})$ and $(\hat \Phi'_{\mathfrak{R}'},\hat \Pi'_{\mathfrak{R}'})$, and a quantum state $\hat \rho$, we can apply techniques from finite-dimensional systems and compute the entanglement between the two subsystems each pair defines. Obviously, we can also consider more complicated  subsystems, each made of an arbitrary  but finite number of independent ``modes''. This strategy can be used to evaluate the entanglement generated during inflation. For instance, one can use the Peres-Horodecki criterion \cite{PeresSeparabilityPhysRevLett,HORODECKI1997333,SimonSeparability2000} to the reduced quantum state describing  the two subsystems, which provides a sufficient condition for separability, and in certain situations can be used to define entanglement quantifiers \cite{serafini2017quantum}. Such a calculation goes beyond the scope of this paper. Nevertheless, we want to emphasize two important aspects of it:\\
	
	(1) Given a ``field observable'' $\hat  \Phi_\mathfrak{R}$, there is freedom in choosing a conjugate momentum. In other words, there are infinitely many operators $\hat \Pi_\mathfrak{R}$ satisfying $ [\hat \Phi_\mathfrak{R},\hat \Pi_\mathfrak{R}]=1$. This well-known fact has important consequences for entanglement, since each choice of  $\hat \Pi_\mathfrak{R}$  gives rise to a {\em different} physical subsystem, and therefore to {\em a different result for the entanglement} with the other subsystem. Hence, it is important to keep in mind that entanglement requires a choice of both field and momentum observables.  
	
	(2) No information about  the conjugate momenta of the primordial perturbations has been extracted from data so far. In fact, such observations are considered nonviable, at least for the accepted models of inflation  (see e.g.\ \cite{Maldacena:2015bha} for a discussion about this point), because information about the momenta must be extracted from the time derivatives of the field, which are exponentially small in the simplest inflationary models. The absence of information about  momentum observables $\hat \Pi_\mathfrak{R}$ precludes us from  making any statement related to entanglement (see Appendix \ref{app:qq} for further discussions).

	Current data inform us only about the statistics of field observables $\hat \Phi_\mathfrak{R}$ at a given time, and reveal that the observed correlations between the field at spatially separated regions of space are stronger than one would expect in the vacuum in Minkowski spacetime.
	The inflationary paradigm can account for these observed correlations, due to the fact that the primordial power spectrum $P_{\mathcal{R}}(k)\propto H^2$ remains almost scale invariant for super-horizon scales, in contrast to the decay $P_{\mathcal{R}}(k)\propto k^2$ one would find in Minkowski spacetime when $k\to 0$.
	But stronger correlations in the field  $\hat \Phi_\mathfrak{R}$ do not necessarily imply stronger entanglement. In fact,  these correlations alone cannot inform us about entanglement, since all the field  observables commute. As emphasized above, to speak about entanglement one needs to involve observables not commuting among themselves.  In the absence of momentum observables, there is no way to check whether the observed correlations come together with any entanglement. All observations so far can be accounted for by a classical theory, as previously emphasized in \cite{Martin:2015qta,ack,Maldacena:2015bha}.  In the absence of non-Gaussianity, all observations can be accounted for by a classical Gaussian stochastic state (i.e., a probability distribution in phase space) with appropriate mean and covariance, since the differences between such a classical state and a quantum Gaussian state are accessible only if non-commuting observables are measured. 
	This is true even in the idealized case considered here, in which we have ignored the effects of the ubiquitous decoherence processes that may have affected the primordial perturbations, e.g.\ due to their interaction with matter and radiation in the universe, as well as to potential self-interactions among different modes.

	\section{Discussion\label{sec:concl}}

	The first investigations of quantum field theory in curved spacetimes focused on FLRW  universes that become asymptotically Minkowskian in the past and future \cite{Parker:1968mv,Parker:1969au,Parker:1971pt}, and interpreted the effects that the expansion of the universe produces on quantum fields in terms of  particles created out of the vacuum. Soon after, it was understood that such a particle interpretation is not available in more realistic FLRW geometries, except in special circumstances. In general, the notion of particles is ambiguous at cosmological scales. The physical and invariant properties of the system are better encoded in field observables. This is indeed the strategy  followed in the context of inflation, where the state of the primordial perturbations at the end of the inflationary era is commonly characterized by the two-point correlation function---the power spectrum. Speaking about particles created during inflation introduces an unnecessary ambiguity, which masks the information that is relevant for observations. 
	
	In this paper, we argue that the question ``does the expansion of the universe squeeze and entangle perturbations with wavenumber $\vec k$ and $-\vec k$?'' suffers from exactly the same ambiguities as the calculation of the creation of particles. The reason is that in order to quantify squeezing and entanglement one needs to define observables associated with wavenumber $\vec k$ and $-\vec k$ and such a definition requires a choice of creation and annihilation operators. 
	More importantly, the ambiguity in the definition of   creation and annihilation operators {\em mixes degrees of freedom with wavenumber $\vec k$ and $-\vec k$} [see Eq.~\eqref{mix}], directly affecting the answer one is trying to find. In other words, there is no unambiguous separation of the degrees of freedom of the field in those associated with $\vec k$ and those with $-\vec k$, contrary to what one could intuitively think---except in special circumstances where additional symmetries are present. In fact, given any homogeneous and isotropic Gaussian state in an FLRW spacetime, there always exists a choice of creation and annihilation variables for which the state contains no particles and no squeezing or entanglement between the $\vec k$ and $-\vec k$ sectors. This is the case in inflation if one adapts the definition of particles to the approximate de Sitter symmetries of the spacetime. Other choices are of course possible ---although, in our opinion,  less desirable (see section~\ref{sec:previous-work}) --- and  lead to different conclusions. The absence of particle creation, generation of squeezing and entanglement for a given choice, however, does not make the state more or less quantum, since the physical properties of the state are insensitive to this ambiguity, as one would have expected. This situation is not very different from what happens in the well-known Unruh effect  \cite{Fulling:1972md,Davies:1974th,Unruh:1976db}: the Minkowski vacuum is made of maximally entangled pairs of particles, as defined by Rindler observers. Obviously, the way these observers perceive the Minkowski vacuum  does not make the state more or less quantum. The invariant information in the state is encoded in its two-point function $\langle 0|\Phi(t,\vec x)\Phi(t',\vec x')|0\rangle$ rather than in its particle content. 
	
	Observations in cosmology are performed in real space and  are independent of the ambiguities appearing in Fourier space, as we have discussed in section~\ref{sec:real}. Therefore, smeared fields in real space are the type of observables one should focus on to quantify the ``quantumness'' of the CMB, as has been recently emphasized e.g.~in \cite{Martin:2021qkg}. Observations of the CMB reveal that the primordial perturbations are correlated at spatial separations more strongly than they would be in the vacuum in Minkowski spacetime. Inflation can indeed account for these strong correlations. But stronger correlations in field observables do not necessarily imply stronger entanglement.  
	Reference \cite{Martin:2021qkg} has studied quantum discord as a way of measuring quantum correlations between space-like separated regions at the end of inflation. However, as  emphasized in \cite{Martin:2021qkg}, quantum discord does not have a clear physical meaning when applied to mixed states (which are unavoidably in this context due to the process of ``tracing out'' the degrees of freedom in other regions of space). It would be desirable to study quantifiers of entanglement with a more transparent physical interpretation in order to understand whether the generation of field correlations in inflation comes together with the production of quantum entanglement in real space. 
	
	Even if one restricts to observables in real space, the computation of entanglement faces an ambiguity: the choice of conjugate momentum of a field observable. This ambiguity is enormous in field theory, and it directly affects any calculation of entanglement. In practice, one could make a choice based on what can be actually observed, but unfortunately current CMB data inform us only about correlations among field observables at a given time, and not about any non-commuting observable which could play the role of a conjugate momentum. Therefore, current data is insufficient to determine whether there is any entanglement associated with the observed correlations. The question ``how much entanglement exists in this quantum state?'' is meaningless unless one has access to non-commuting observables.  Consequently, there is nothing genuinely quantum in observations made so far, in the sense that all observations can be  satisfactorily accounted for by a classical theory and a stochastic classical state (a probability distribution function in phase space), as previously emphasized in \cite{Martin:2015qta,ack}.  
	
	In this article, we  have focused on Gaussian states and ignored possible deviations from Gaussianity, motivated by the fact that observations to date have been unable to reveal  deviations from Gaussianity in the primordial perturbations. Gaussian states are easy to manipulate and to perform calculations with, but they are also limited in the content of genuinely quantum information one can extract from them. Nonetheless, it may be that the primordial perturbations are non-Gaussian. Observation of non-Gaussianities (see e.g.~\cite{Baumann:2021ykm} for a possible way using the large scale structure) could make the task of finding whether the origin of the primordial perturbations is quantum easier, as recently advocated in \cite{Green:2020whw}. 
	
	We have also ignored the effects of decoherence ---either produced by interactions of the primordial  perturbations with matter and radiation soon after the end of inflation, or from self-interactions--- in washing away quantum aspects, since decoherence would only aggravate the possibility of observing any quantum trace from the mechanism that generated the perturbations. 
	Finally, we have not considered potential deviations from the standard theory of quantum measurement, which could play a role in cosmology and account for the absence of the exact (rather than statistical) rotational symmetry of CMB anisotropies, as described in \cite{Sudarsky:2009za,Canate:2012nwv,Leon:2015hwa,Leon:2016ysi,Leon:2017yna,Bengochea:2020efe}. The effects of a potential spontaneous collapse process of the wave function could also add to the classicality of primordial perturbations, as discussed in \cite{Berjon:2020vdv}.

	\section*{Acknowledgments}
	We have benefited from  discussions with A. Ashtekar, A. Brady, A. Delhom, D. Kranas, J. Pullin, J. Olmedo, S. Suddhasattwa and M. Wilde. This work is supported by the NSF grant PHY-2110273, and from the Hearne Institute for Theoretical Physics.

	\appendix


	\section{One-mode squeezing in the Poincar\'e patch of de Sitter space \label{app:single-mode-sqz}}
	
	In addition to two-mode squeezing and entanglement between the degrees of freedom of a scalar field  associated with wavenumbers $\vec k$ and $-\vec k$ discussed in section \ref{sec:scalar}, the question of one-mode squeezing and its relation to the quantumness of cosmological perturbations has been also discussed in the literature \cite{ack}. In this section, we discuss the generation of single-mode squeezing in FLRW geometries in general, and inflationary spacetimes in particular, emphasizing  the ambiguities it involves. We investigate these questions for the two strategies introduced in  section \ref{sec:scalar} to define quadrature-pairs associated with a wavenumber $\vec k$.
	
	For the intermediate calculations in this section, it will be convenient to  ``put the universe in a box'' of coordinate volume $V_0$, with  $V_0$ finite but arbitrary large. This helps to avoid the mathematical inconvenience of having modes normalized to the Dirac delta distribution, and will make the arguments below more transparent. 
	Physical quantities do not depend on $V_0$, and we can send it to infinity at the end of the calculation. Mathematically, sending  $V_0\to \infty$ is equivalent to replacing $V_0$ by $(2\pi)^3$ in all calculations below.\\

	{\bf Strategy  1:}  Quadrature-pairs defined from the Fourier modes of the field.  As discussed above, from the (non-Hermitian) Fourier modes of the field
	\be \label{phip2} \phi_{\vec k}=\frac{1}{V_0} \int d^3x\, e^{-i\, \vec k\cdot \vec x}\, \hat \Phi(\vec x)\, ,  \hspace{1cm}  \pi_{\vec k}= \int d^3x\, e^{-i\, \vec k\cdot \vec x}\, \hat \Pi(\vec x) \, , \ee
	one can define two Hermitian quadrature-pairs 
	\bea \label{eq:hermitian} \hat \phi^{({\rm R})}_{\vec k}&\equiv&\sqrt{k\, V_0}\, a\, \frac{1}{\sqrt 2}\, (\hat \phi_{\vec k} +\hat \phi_{-\vec k} ) \, , \ \ \ \  \hat \phi^{({\rm I})}_{\vec k}\equiv \sqrt{k\, V_0}\, a\, \frac{-i}{\sqrt 2}\, (\hat \phi_{\vec k} -\hat  \phi_{-\vec k} )\, , \nonumber \\ 
	\hat \pi^{({\rm R})}_{\vec k}&\equiv&\frac{1}{\sqrt{k\, V_0}\, a}\,  \frac{1}{\sqrt 2}\,(\hat \pi_{\vec k} +\hat  \pi_{-\vec k} )\, , \ \ \ \ \hat  \pi^{({\rm I})}_{\vec k}\equiv \frac{-i}{\sqrt{k\, V_0}\, a}\,  \frac{1}{\sqrt 2}\, (\hat \pi_{\vec k} -\hat  \pi_{-\vec k} )\,  \eea
	for $\vec k \in  \frac{2\pi}{V_0^{1/3}} \, \mathbb{Z}^3_{(+)}$, where $\mathbb{Z}^3_{(+)}$ is defined in footnote \ref{Z+}. The pre-factors $\sqrt{k\, V_0}\, a$ and $\frac{1}{\sqrt{k\, V_0}\, a}$ in the fields and momenta were not introduced in section \ref{sec:scalar}, but they will be important in the following discussion. They are motivated by: (i) The two elements of each quadrature-pair need to have the same dimensions---square root of action---in order to  define squeezing (see section \ref{quadsqz}). (ii) Only coordinate-independent, physical quantities enter in their definition---e.g., $\sqrt{k\, V_0}\, a=\sqrt{\frac{k}{a}}\sqrt{V_0\, a^3}$, which depends on the physical wavenumber $k_{\rm ph}=\frac{k}{a}$ and the physical volume $a^3V_0$, and therefore do not change under mere re-scaling of the coordinates $\vec x \to \alpha\, \vec x$, with $\alpha$ a real number. Restricting to coordinate invariant observables will guarantee that our results below have direct physical meaning and cannot be attributed to coordinate artifacts. 
	
	The quadrature-pairs (\ref{eq:hermitian}) are defined using only structures present in the FLRW geometry and do not require any additional choice. In particular, they do not require a choice of creation and annihilation variables, neither at the initial nor at the final time. They are, therefore, free from extra ambiguities and can be used to define the generation of single-mode squeezing in an invariant manner.  
	
	Let us now focus on the Poincar\'e patch of de Sitter (PdS) spacetime, and assume the field is in the BD vacuum at time $\eta_0$ (see Appendix \ref{app:BD}) . The calculations below are more straightforward in the Heisenberg picture. We will denote by $|{\rm BD}\rangle$ the BD vacuum and by $\hat A_{\vec k}$ the operators annihilating them. It is then convenient to expand  $ \hat \phi_{\vec k}(\eta)$ in terms of $\hat A_{\vec k}$ and its Hermitian conjugate:
	%
	\be  \hat \phi_{\vec k}(\eta)=e^{{\rm BD}}_k(\eta)\, \hat A_{\vec k}+e^{{\rm BD}\, *}_k(\eta)\, \hat A^{\dagger}_{-\vec k}\,  \ee
	where the modes $e^{{\rm BD}}_k(\eta)$ are written in expression (\ref{BDmodes}). 
	From this,  we obtain a representation for the Hermitian fields
	\bea  \hat \phi^{(R)}_{\vec k}(\eta)&=&\sqrt{V_0 \, k}\, a(\eta)\,  \Big( e^{{\rm BD}}_k(\eta)\, \hat A^{(R)}_{\vec k}+e^{{\rm BD}\, *}_k(\eta)\, \hat A^{(R)\, \dagger}_{\vec k}\Big)\, , \\    \hat \phi^{(I)}_{\vec k}(\eta)&=&\sqrt{V_0 \, k}\, a(\eta)\,  \Big( e^{{\rm BD}}_k(\eta)\, \hat A^{(I)}_{\vec k}+e^{{\rm BD}\, *}_k(\eta)\, \hat A^{(I)\, \dagger}_{\vec k}\Big) \, ,\eea
	where we have defined $\hat A^{(R)}_{\vec k}\equiv \frac{1}{\sqrt{2}}\, ( \hat A_{\vec k}+ \hat A_{-\vec k})$ and $\hat A^{(I)}_{\vec k}\equiv \frac{-i}{\sqrt{2}}\, ( \hat A_{\vec k}- \hat A_{-\vec k})$. One can easily check that $[\hat A^{(R)}_{\vec k},\hat A^{(R)\, \dagger}_{\vec k'}]=\delta_{\vec k,\vec k'}=[\hat A^{(I)}_{\vec k},\hat A^{(I)\, \dagger}_{\vec k'}]$ and $[\hat A^{(R)}_{\vec k},\hat A^{(I)}_{\vec k'}]=0=[\hat A^{(R)}_{\vec k},\hat A^{(I)\, \dagger}_{\vec k'}]$. 
	From this, it is straightforward to obtain that the cross-correlations between  $\hat \phi^{(R)}_{\vec k}$ and $ \hat \phi^{(I)}_{\vec k}$ vanish at any time: $\langle  {\rm BD}|  \hat \phi^{(R)}_{\vec k}(\eta)\, \hat \phi^{(I)}_{\vec k}(\eta)|{\rm BD} \rangle =0$. Furthermore, the variances of $\hat \phi^{({\rm R})}_{\vec k}$ and $\hat \phi^{({\rm I})}_{\vec k}$ are equal to each other and can be written in terms of the non-Hermitian fields $\hat \phi_{\vec k}$ 
	\be  (\Delta  \hat \phi^{({\rm R})}_{\vec k})^2= (\Delta  \hat \phi^{({\rm I})}_{\vec k})^2=V_0\, k\, a^2\, \langle \hat \phi_{\vec k}\hat \phi_{-\vec k}\rangle\,  .\ee
	This justifies why it is standard to focus attention on $\langle \hat \phi_{\vec k}\hat \phi_{-\vec k}\rangle$, even though $ \hat \phi_{\vec k}$ is not an observable, since the correlations among the Hermitian fields can all easily obtained from it. 
	The same argument applies to the momenta: 
	\be   (\Delta  \hat \pi^{({\rm R})}_{\vec k})^2= (\Delta  \hat \pi^{({\rm I})}_{\vec k})^2=\frac{1}{V_0\, k\, a^2}\, \langle \hat \pi_{\vec k}\hat \pi_{-\vec k}\rangle\,  .\ee
	Let us use the  Bunch-Davies form for the modes of a massless, minimally coupled field, $m=0=\xi$, as this case is often discussed in the literature because it is closer to the study of cosmological perturbations: 
	\be e^{{\rm BD}}_k(\eta)=\frac{1}{a(\eta)\sqrt{V_0}}\,\frac{1}{\sqrt{2 k}}\, \left(1-i\, \frac{1}{\eta\, k}\right)\,  e^{-i\, k\, \eta}\, . \ee
	Then,  we obtain 
	\be \label{PS} \langle \hat \phi_{\vec k}\hat \phi_{-\vec k}\rangle=\frac{1}{V_0} \frac{1}{a(\eta)^2\, 2\, k}\, \left(1+\frac{H^2}{k_{\rm ph}^2(\eta)}\right)\, ,\ee
	so that 
	\bea  (\Delta  \hat \phi^{({\rm R})}_{\vec k})^2&=& (\Delta  \hat \phi^{({\rm I})}_{\vec k})^2=\frac{1}{2}\,  \left(1+\frac{H^2}{k_{\rm ph}^2(\eta)}\right)\nonumber \\
	(\Delta  \hat \pi^{({\rm R})}_{\vec k})^2&=& (\Delta  \hat \pi^{({\rm I})}_{\vec k})^2=\frac{1}{2}. \eea
	Note that: (i) These variances only depend on physical quantities, namely the Hubble rate $H$ and the physical wavenumber $k_{\rm ph}=\frac{k}{a}$. In particular, they do not depend on the volume $V_0$, which can be sent to infinity if desired.\footnote{Independently of the value of $V_0$, this expression blows up  for the zero mode. This corresponds to the well-known infrared divergences of the Bunch-Davies vacuum for massless fields. In cosmology, one assumes that the state is modified for very infrared scales in a way that makes it infrared finite---at the expense of breaking the PdS invariance. Since the modification is irrelevant for observations, because it only involves extreme infrared scales, there is no need to specify it.} (ii) In the limit $H\to 0$ one recovers the result expected for Minkowski spacetime and the Minkowski vacuum, namely all four variances are equal to $\frac{1}{2}$. (iii) While the variances of the momenta are time-independent, the variances of the fields {\em grow in time}. Therefore, the evolution does not squeeze these two quadrature-pairs. 
	
	One could be tempted to focus instead on the non-Hermitian field $\hat \phi_{\vec k}$ and argue that, because \eqref{PS} decreases in time due to the presence of the pre-factor $a^{-2}$,  single-mode squeezing occurs during inflation. We do not support this argument because  $\hat \phi_{\vec k}$ is not an observable (since it is a non-Hermitian operator) and, as we just showed, a natural way of interpreting \eqref{PS} in terms of Hermitian quadratures shows that the later do not get squeezed. 

	We conclude that there is no reason to support that single-mode squeezing happens during inflation for the Fourier modes associated with wavenumber $\vec k$, as long as one restricts attention to physical quantities. \\

	{\bf Strategy  2:} A second option is to define quadratures, as we did in the discussion of two-mode squeezing, from annihilation and creation variables
	\be \hat X_{\vec k}=\frac{1}{\sqrt{2}} \, (\hat A_{\vec k}+\hat A_{\vec k}^{\dagger})\, , \qquad \hat P_{\vec k}=-\, \frac{i}{\sqrt{2}} \, (\hat A_{\vec k}-\hat A_{\vec k}^{\dagger})\, . \ee
	This definition, however, requires a choice for $\hat A_{\vec k}$ and $\hat A_{\vec k}^{\dagger}$, and consequently the result depends on this choice. As before, in  PdS   one can use the spacetime symmetries to single out the one-parameter family of quadrature-pairs written in \eqref{BDq}. The discussion of single-mode squeezing becomes identical to the discussion of two-mode squeezing in section \ref{sec:scalar}: if the system is prepared in the BD vacuum at time $\eta_0$ (in Schr\"odinger's picture), it will evolve to the BD vacuum at time $\eta_1>\eta_0$. The final state will look physically identical to an observer at $\eta_1$ as the initial state did for an observer at $\eta_0$. Hence, there is no generation of single-mode squeezing during inflation if one uses the spacetime symmetries of PdS to remove the ambiguities.

	\section{Relation between squeezing, entanglement, and quantumness of Gaussian states\label{app:qq}}
	
	Several ways to identify and quantify the ``quantumness'' of Gaussian states are frequently described in textbooks (see, e.g., \cite{serafini2017quantum}). Squeezing and entanglement are two such examples, but there are also others like the $P$-function. In the main body of this paper, we showed that squeezing and entanglement are not intrinsic properties of the state, since they depend on  a  choice of quadratures. This also applies to the alternative methods of describing the ``quantumness'' of Gaussian states. We will elaborate upon this here, first, for the simplified setting of quantum mechanics and next, for quantum field theory on time-dependent  spacetimes. 
	
	\subsection{P-function and entanglement in quantum mechanics}
	Perhaps the most commonly used ``quantumness'' criterion is the so-called $P$-function, from which one  characterizes a Gaussian state (or any state in general) as quantum if its $P$-function takes negative values \cite{gerry2005introductory}. For a time-independent simple oscillator, squeezed states are examples of states with negative $P$-functions, while vacuum and coherent states in general are not. Regarding separability, a sufficient criterion to detect its presence is by paying attention to the (ordinary) eigenvalues of the covariance matrix: if 
	all the ordinary eigenvalues of $\sigma$ are equal or larger than one, then a Gaussian state is separable across any bipartition of the quadrature-pairs. Do these widely used definitions provide an invariant way of characterizing quantumness and separability? The answer is no, since a close inspection reveals that neither of these properties of the $P$-function and the ordinary eigenvalues of $\sigma$ are  symplectic invariant. On the one hand, the $P$-function rests on a choice of annihilation and creation variables (or equivalently, on a choice of quadrature-pairs). This is clear from its definition, and also from the fact that its characteristic function is a generating function for {\em normal-ordered products}. Since normal-order requires a choice  of annihilation and creation variables, the $P$-function also does. Hence, the $P$-function measures if a state is squeezed {\em relative} to the vacuum singled out by the choice of annihilation operators made to define it. Similarly, regarding separability, as already mentioned in footnote \ref{fn1}, the ordinary eigenvalues of $\sigma$ are not symplectic invariant, because not all symplectic transformations are orthogonal, and hence the eigenvalues only provide information about the basis of quadratures one is using to write $\sigma$. 
	
	\subsection{Wigner function in quantum mechanics}
	The Wigner function $\rho_W(x_i,p_i)$ {\em does} provide a symplectic covariant way of describing the properties of a Gaussian state, since its definition does not require any additional structure or choice. As is well-known (see e.g. \cite{gerry2005introductory}), the Wigner function of every Gaussian state (pure or mixed) is a Gaussian probability density function (PDF) in the classical phase space, with mean and variance given by the first and second moments of the quantum state, $\vec \mu$ and $\sigma$. This PDF has three important properties: (i) It is positive across the entire phase space. (ii) For quadratic Hamiltonians, its time evolution satisfies Liouville's classical equation of motion $\frac{d}{dt} \rho_W=-\{\rho_W,H\}$, where the curly brackets represent Poisson brackets. (This is a consequence of the fact that,  for quadratic Hamiltonians, the classical evolution completely determines the quantum dynamics, as discussed in section \ref{dyn}). Hence, the Wigner function is a bona fide ``classical mixed state'', i.e., a stochastic classical state. These properties are symplectic invariant and true for {\em all Gaussian states}, pure or mixed. (iii) The classical average of any polynomial in $x_i$ and $p_i$ with respect to the Wigner function  $\rho_W(x_i,p_i)$ agrees {\em exactly} with the quantum expectation value of the symmetrically ordered version of the polynomial. Hence, if we restrict to symmetrically ordered functions of $\hat x_i$ and $\hat p_i$, Gaussian states and quadratic Hamiltonians, we can completely forget about the quantum formalism, and obtain all physical predictions by working on a classical theory with a stochastic classical state  $\rho_W(x_i,p_i)$ --- as  already emphasized in the context of inflation in \cite{Martin:2015qta,ack}. In this precise and invariant sense, all Gaussian states are ``the most classical  states in the quantum theory''. The Wigner function does not distinguish between ground, squeezed, thermal, coherent or other types of Gaussian states. 
	
	\subsection{Squeezing, entanglement and quantumness in field theory}
	For field theory in time-dependent spacetimes,  a preferred set of quadratures is not available in general. In fact, one can always find quadrature-pairs containing arbitrary amounts of squeezing or entanglement for all Gaussian states. One is therefore forced to put all Gaussian states on equal footing  and the labels  ``squeezed'' and ``entangled''  become  mere conventions, based on a choice of quadrature-pairs. All Gaussian states are equally quantum (or equivalently, equally classical). In special circumstances when physically preferred quadrature-pairs exist, one commonly links the degree of quantumness of states to them---recall that a preferred quadrature-pair is equivalent to having a preferred set of annihilation and creation variables and vacuum state. This is what we do, for instance, in Minkowski spacetime, for which the Minkowski vacuum and coherent states built from it are considered the most classical states of the theory. Similarly,  one uses the name ``squeezed states'' for states presenting squeezing or entanglement for the preferred quadratures-pairs and regards them as highly quantum states. But in time-dependent spacetimes, where one does not have preferred quadrature-pairs nor a preferred vacuum state, the difference between coherent and squeezed states dilutes. {\em All Gaussian states are on equal footing regarding their quantumness or classicality.} 
	
	
	Note that one can use the Wigner function to evaluate the quantumness of any state in a symplectic-covariant way, since the definition  of the Wigner function does not require a choice of quadrature-pairs or basis in the classical phase space (in contrast to, for instance, the $P$-function). As for finite-dimensional systems, the Wigner function is positive-definite across the classical phase space for all Gaussian states, and it satisfies the classical equations of motion for quadratic Hamiltonians. Moreover, like for finite-dimensional systems, the Wigner function puts all Gaussian states on equal footing regarding their quantumness or classicality.
	
	
	\section{The Bunch-Davies vacuum in Schr\"odinger's picture \label{app:BD}}
	
	This section provides some details omitted in subsection \ref{subsec:PdS} about the definition of the so-called Bunch-Davies vacuum in the Poincar\'e patch of de Sitter spacetime (PdS). In particular, we discuss how this state is defined in Schr\"odinger's picture, and prove some results which do not commonly arise if one works in Heisenberg's picture and which we explicitly used in subsection \ref{subsec:PdS}.

	\subsection{How to define a Fock-vacuum in Schr\"odinger's picture}
	
	In quantum field theory on curved spacetimes, there are infinitely many states which can play the role of the Fock vacuum. They are commonly referred to as quasi-free states, and can be characterized as pure Gaussian states with zero mean (i.e.\ $\langle \hat \Phi(\vec x)\rangle =0=\langle \hat \Pi(\vec x)\rangle$). These states are fully characterized by their covariance matrix $\sigma(\vec x,\vec x')$ specified at a given instant. Therefore, there are as many Fock vacua as inequivalent covariant matrices describing pure Gaussian states. In practice, one can select a Fock vacuum by choosing a basis in the complexified classical phase space, as follows. 
	
	Let $\Gamma$ be the classical phase space  of the field theory, and let $\Gamma_{\mathbb{C}}$ be its complexification---obtained by taking all possible linear combinations of elements of  $\Gamma$ with complex coefficients. The classical symplectic structure allows us to define a ``product'' in $\Gamma_{\mathbb{C}}$: Given two elements in  $\Gamma_{\mathbb{C}}$,   $\gamma(\vec x)=(e(\vec x), p(\vec x))$ and $\tilde \gamma(\vec x)=(\tilde e(\vec x), \tilde p(\vec x))$, their product is\footnote{As in Appendix \ref{app:single-mode-sqz}, we will work with a universe of finite volume $V_0$, and send $V_0\to \infty$ only at the end of the calculations. This avoids the mathematical awkwardness of working with mode functions normalized to the Dirac delta distribution.} 
	\be \langle \gamma,\tilde \gamma\rangle=i\, \int \frac{d^3x}{V_0} \, ( e(\vec x)^*\, \tilde p(\vec x)- p(\vec x)^*\, \tilde e(\vec x))\, . \ee
	This product satisfies all the properties of a Hermitian  inner product in  $\Gamma_{\mathbb{C}}$, except one---it is not positive definite. Yet one can always decompose $\Gamma_{\mathbb{C}}$ into a direct sum of a subspace $\Gamma_{+}$ on which the product is positive definite, and its orthogonal complement $\Gamma_{-}$ (which happens to be the subspace complex conjugated to $\Gamma_{+}$,  and on which the product is guaranteed to be negative definite).  The important statement is that such splitting of $\Gamma_{\mathbb{C}}$ in subspaces of positive and negative norm vectors, $\Gamma_{\mathbb{C}}=\Gamma_{+} \oplus \Gamma_{-}$,  is equivalent to a choice of Fock vacuum. The relation is as follows: given an orthonormal basis in  $\Gamma_{+}$, $\{\gamma_{\alpha}(\vec x)\}_{\alpha=1}^{\infty}$, the covariance matrix of the Fock vacuum is 
	\be \sigma^{ij}(\vec x,\vec x')=\sum_{\alpha} \gamma^i_{\alpha}(\vec x)\gamma^{j\, *}_{\alpha}(\vec x')+\gamma^j_{\alpha}(\vec x)\gamma^{i\, *}_{\alpha}(\vec x')\, , \ee
	where $i,j=1,2$ label the field- and conjugate momentum-components of phase space elements $\gamma^i_{\alpha}$.\footnote{Because we have not smeared out the fields, $\sigma^{ij}(\vec x,\vec x')$ must be understood in the distributional sense.} It is straightforward to show that this covariant matrix does not depend on the concrete basis we use within $\Gamma_{+}$. Hence, there is a one-to-one correspondence  between positive norm subspaces  $\Gamma_+$ and covariance matrices of pure Gaussian states. 
	
	Obviously, since there are infinitely many different splittings $\Gamma_{\mathbb{C}}=\Gamma_{+} \oplus \Gamma_{-}$, there is a huge ambiguity in the definition of a Fock vacuum. If additional symmetries are present, one can use these to narrow down the ambiguity. Below, we will study the way the symmetries of PdS  affect this ambiguity. At the technical level, we will do this by studying the way a positive-norm subspace  $\Gamma_+$ changes under the transformations within the PdS symmetry group, by applying the transformations to a basis. The Fock vacuum is  invariant under a group of transformations if $\Gamma_+$ is left invariant. 
	
	A more direct, although equivalent way of understanding the symmetries of the vacuum is by looking at the symmetries of the covariant matrix itself. However, in practice it is more convenient to work directly with a basis in $\Gamma_+$ to understand the restriction the symmetries impose on the choice of vacuum. We follow the latter route in this section.

	\subsection{The isometries of PdS}
	
	The de Sitter group in four spacetime dimensions has ten independent Killing vectors fields, and all of them, locally, are isometries of PdS. But since  the PdS  is only a portion of de Sitter space, not all these transformations are global isometries of PdS. Only the subgroup of the de Sitter group which leaves the Poincar\'e patch invariant describes the global isometries of  PdS (see  \cite[Sec.~IV~C]{abk1}). The global isometries of the PdS form a seven-dimensional group, generated by three spatial translations and three rotations (these are common to all spatially flat FLRW spacetimes), and one additional isometry defined by the Killing vector field
	\be K^{\mu}=-H\, \eta\, \partial^{\mu}_{\eta}-H\, x\, \partial^{\mu}_{x}-H\, y\, \partial^{\mu}_{y}-H\, z\, \partial^{\mu}_{z}\, . \ee
	The orbits of this vector field  passing through an arbitrary point $(\eta_0,x_0,y_0,z_0)$ are  the curves $c^{\mu}(\lambda)=e^{-H\, \lambda}\, (\eta_0,x_0,y_0,z_0)$. We see that these orbits combine a translation forward in time $\eta_0\to e^{-H\, \lambda}\,  \eta_0$ (in cosmic time this reads $t_0\to t_0+\lambda)$ with a spatial contraction $\vec x_0\to e^{-H\, \lambda}\, \vec x_0$. These transformations leave the metric invariant, since this spatial contraction  exactly compensates for the cosmic expansion occurring  during the time  translation $t_0\to t_0+\lambda$.

	\subsection{Bunch-Davies vacuum at instant \texorpdfstring{$\eta_0$}{eta\_0}}
	
	It will be convenient to consider the following elements of $\Gamma_{\mathbb{C}}$:
	\be \label{BDmodesM} \gamma^{\rm  BD}_{\vec k}(\eta_0, \vec x)\equiv  \begin{pmatrix} e^{\rm BD}_{\vec k}(\eta_0, \vec x)\\ p^{\rm BD}_{\vec k}(\eta_0, \vec x) \end{pmatrix}=  \begin{pmatrix} \sqrt{\frac{-\pi\, \eta_0}{V_0\, 4\, a^2(\eta_0)}}\, H^{(1)}_{\mu}(-k\, \eta_0)\, e^{i\, \vec k \cdot \vec x} \\ V_0\, a^2(\eta_0)\frac{d}{d\eta}|_{\eta_0} e^{\rm BD}_{\vec k}(\eta, \vec x)   \end{pmatrix}\ee
	where  $ H^{(1)}_{\mu}(-k\, \eta_0) $ is a Hankel function with index  $\mu^2=\frac{9}{4}-\frac{m^2}{H^2}-12 \, \xi$.  One can easily check that these modes are orthonormal, $\langle  \gamma^{\rm  BD}_{\vec k}, \gamma^{\rm  BD}_{\vec k'}\rangle =\delta_{\vec k,\vec k'}$. The set $\{ \gamma^{\rm  BD}_{\vec k}(\eta_0, \vec x)\}$, for all $\vec k$, defines a vacuum, and we will show in the following that it is the only vacuum state that is both PdS invariant and Hadamard at $\eta_0$.
	
	Consider   a more general family of modes, defined as 
	\be \label{basis} \gamma_{\vec k}= \alpha_{\vec k}\, \gamma^{\rm  BD}_{\vec k}(\eta_0, \vec x)+\beta_{\vec k}\, \gamma^{{\rm  BD}\, *}_{-\vec k}(\eta_0, \vec x)\, , \ee
	with $\alpha_{\vec k}$ and $\beta_{\vec k}$ arbitrary complex numbers satisfying $|\alpha_{\vec k}|^2-|\beta_{\vec k}|^2=1$. Given a choice for  $\alpha_{\vec k}$ and $\beta_{\vec k}$ for all values of $\vec k$, the set  $\{ \gamma_{\vec k}(\vec x)\}_{\vec k}$ defines a Fock vacuum with the following properties:
	
	\begin{theorem}
		The vacua defined by any of the sets $\{ \gamma_{\vec k}(\vec x)\}_{\vec k}$ are all invariant under translations. 
		
	\end{theorem}
	
	\begin{proof}
		The proof is rather trivial; because  the $\vec x$-dependence in $ \gamma_{\vec k}(\vec x)$ is of the  simple form $e^{i\vec k\cdot \vec x}$, a translation $\vec x\to \vec x+\vec \lambda$ changes 
		$\gamma_{\vec k}(\vec x)\to e^{i\vec k\cdot \vec \lambda}\,   \gamma_{\vec k}(\vec x)$, and obviously these phases leave the vector space  $\Gamma_+={\rm span}\{ \gamma_{\vec k}(\vec x)\}_{\vec k}$ invariant. 
	\end{proof}

	\begin{theorem}
		The vacua defined by any of the sets $\{ \gamma_{\vec k}(\vec x)\}_{\vec k}$ are invariant under rotations if and only if $\alpha_{\vec k}=\alpha_k$ and $\beta_{\vec k}=\beta_k$ for all $\vec k$, that is, if these coefficients do not depend on the direction of $\vec k$.   
	\end{theorem}

	\begin{proof}
		Under a rotation $R$, $\gamma_{\vec k}(\vec x)\to \gamma_{\vec k}(R\cdot \vec x)=  \alpha_{\vec k}\, \gamma^{\rm  BD}_{R^{\top}\cdot  \vec k}(\eta_0, \vec x)+\beta_{\vec k}\, \gamma^{{\rm  BD}\, *}_{-R^{\top}\cdot  \vec k}(\eta_0, \vec x) $.
		
		Then, if $\alpha_{\vec k}=\alpha_k$ and $\beta_{\vec k}=\beta_k$ the transformed modes are equal to $\gamma_{R^{\top}\cdot \vec k}(\vec x)$, and consequently the  vector space  $\Gamma_+={\rm span}\{ \gamma_{\vec k}(\vec x)\}_{\vec k}$ is left invariant. 
		
		Conversely, if the vector space  $\Gamma_+={\rm span}\{ \gamma_{\vec k}(\vec x)\}_{\vec k}$  is invariant under rotations, then $ \alpha_{\vec k}\, \gamma^{\rm  BD}_{R^{\top}\cdot  \vec k}(\eta_0, \vec x)+\beta_{\vec k}\, \gamma^{{\rm  BD}\, *}_{-R^{\top}\cdot  \vec k}(\eta_0, \vec x)$ must belong to it, for any rotation $R$ and for all $\vec k$. This implies that there must exist some complex coefficient $\lambda_{\vec k \vec k'}$ satisfying 
		\be  \label{albet} \alpha_{\vec k}\, \gamma^{\rm  BD}_{R^{\top}\cdot  \vec k}(\eta_0, \vec x)+\beta_{\vec k}\, \gamma^{{\rm  BD}\, *}_{-R^{\top}\cdot  \vec k}(\eta_0, \vec x)=\sum_{\vec k'}\lambda_{\vec k \vec k'}\, \gamma_{\vec k'}(\vec x)\, . \ee
		Using (\ref{basis})  and the orthonormality of $\gamma^{\rm  BD}_{\vec k}$, we deduce that $\lambda_{\vec k \vec k'}$ must be of the form $\lambda_{\vec k \vec k'}=\lambda_{\vec k'}\, \delta_{\vec k',\,  R^{\top}\cdot \vec k}$. With this, equation (\ref{albet}) further implies that $\alpha_{\vec k}=\lambda_{R^{\top}\cdot  \vec k}\ \alpha_{R^{\top}\cdot  \vec k}$ and $\beta_{\vec k}=\lambda_{R^{\top}\cdot  \vec k}\, \beta_{R^{\top}\cdot  \vec k}$ for all rotatins $R$ and for all $\vec k$, which in turn implies $ \lambda_{ \vec k}=1$, $\alpha_{\vec k}=\alpha_{k}$ and $\beta_{\vec k}=\beta_{k}$ for all $\vec k$.
	\end{proof}

	What about invariance under the orbits of $K^{\mu}$? These transformations combine a time translation $\eta_0\to e^{-H\, \lambda}\,  \eta_0$ and a contraction $\vec x\to e^{-H\, \lambda}\, \vec x$, hence:
	\be 
	\gamma_{\vec k}(\vec x) \to \alpha_{ k}\, \gamma^{\rm  BD}_{\vec k}(e^{-H\, \lambda}\, \eta_0, e^{-H\, \lambda}\, \vec x)+\beta_{ k}\, \gamma^{{\rm  BD}\, *}_{-\vec k}(e^{-H\, \lambda}\, \eta_0, e^{-H\, \lambda}\, \vec x)\, . 
	\ee
	Using the analytical form of $\gamma^{\rm  BD}_{\vec k}$ given in (\ref{BDmodesM}), and the fact that the scale factor is  $a(\eta)=-1/(H\eta)$  in PdS spacetimes, with $H$ a constant, we have that $\gamma^{\rm  BD}_{\vec k}(e^{-H\, \lambda}\, \eta_0, e^{-H\, \lambda}\, \vec x)= \gamma^{\rm  BD}_{\vec k'}(\eta_0, \vec x)$, with $\vec k'=e^{-H\, \lambda}\, \vec k$. In other words, the effect of a $K^{\mu}$-transformation is simply to change $\vec k\to\vec k'= e^{-H\, \lambda}\, \vec k$. 
	With this, we find that along the orbits of $K^{\mu}$, $ \gamma_{\vec k}(\vec x)$ transforms as  
	\be  \gamma_{\vec k}(\vec x) \to \alpha_{k}\, \gamma^{\rm  BD}_{\vec k'}(\eta_0, \vec x)+\beta_{k}\, \gamma^{{\rm  BD}\, *}_{- \vec k'}(\eta_0, \vec x) \, , \ee
	with $\vec k'=e^{-H\, \lambda}\,\vec k$. The transformed mode belongs to\footnote{Because the transformations generated by $K^{\mu}$ change the volume $V_0$, the conventions regarding $V_0$ in equations \eqref{phip2} and (\ref{BDmodesM}) must be appropriately chosen, in such a way that Hamilton's equation remain invariant under the transformation. If this is not the case, one cannot blindly compare the phase space elements $\gamma(\vec x)$ before and after the transformation, since they represent initial data for different equations. We have made a choice of factors $V_0$ which makes the comparison meaningful. This remark would be unnecessary  had we decided to work in the covariant phase space, although other aspects of our discussion would be more obscure.} /$\Gamma_+={\rm span}\{ \gamma_{\vec k}(\vec x)\}_{\vec k}$ if and only if $\alpha_k= \alpha_{k'}$ and  $\beta_k= \beta_{k'}$ for all $\lambda$ and all $k$. This implies that $\alpha_k$ and $\beta_k$ must   be independent of $k$. 
	
	Hence, the family of vacua defined from the set $\{\gamma_{\vec  k}\}_{\vec k}$, with $ \gamma_{\vec  k}=\alpha\,   \gamma^{\rm  BD}_{\vec k}(\eta_0, \vec x)+\beta\,   \gamma^{\rm  BD\, *}_{\vec k}(\eta_0, \vec x)$, and $\alpha$ and $\beta$ $k$-independent complex numbers satisfying $|\alpha|^2-|\beta|^2=1$, is invariant under the symmetries of the PdS spacetime. This is the family of the so-called $\alpha$-vacua. 
	
	On the other hand, the Hadamard condition imposes that the  modes $\gamma_{\vec  k}$ defining the vacuum must approach positive frequency modes $e^{-i k\, \eta}$ in the limit $k\to \infty$. This imposes an additional condition on the coefficients, namely that $\beta \to 0$ and $|\alpha|\to 1$ as $k\to \infty$.  This implies that, among all $\alpha$-vacua at instant $\eta_0$, there is only one which is a Hadamard state, namely the one for which $\beta =0$, or in other words, the Fock vacuum defined from the modes $\gamma^{\rm BD}_{\vec k}(\eta_0, \vec x)$ themselves. 
	This is the so-called Bunch-Davies vacuum at $\eta_0$ \cite{Chernikov:1968zm,Tagirov:1972vv,Bunch:1978yq}, which we denote as $|{\rm BD},\eta_0\rangle$. 
	
	In the cosmology literature it is  common to use the name ``Bunch-Davies vacuum'' in a different way, namely to refer to any state that looks like the Minkowski vacuum at short distances. We emphasize that such condition is already captured by the Hadamard condition in a mathematically precise manner, and it does not single out a unique state. It is better to reserve the name Bunch-Davies vacuum for the unique state that is PdS invariant and Hadamard, as originally investigated in  \cite{Chernikov:1968zm,Tagirov:1972vv,Bunch:1978yq}.

	\subsection{Comparing Bunch-Davies vacua at different times in Schr\"odinger's picture}
	
	In the previous section, we reached the conclusion that  $|{\rm BD},\eta_0\rangle$ is the only Hadamard state invariant under the PdS isometries. It is defined from the positive-norm subspace $\Gamma_+={\rm span}\{\gamma^{\rm BD}_{\vec k}(\eta_0, \vec x)\}_{\vec k}$. If we were to repeat the construction at a different time, $\eta_1$,  we would proceed similarly but using instead the modes $\gamma^{\rm BD}_{\vec k}(\eta_1, \vec x)$. Since these  are different phase space elements than  $\gamma^{\rm BD}_{\vec k}(\eta_0, \vec x)$, they potentially define a different state, which we will denote as $|{\rm BD},\eta_1\rangle$, and which is guaranteed, by construction, to be Hadamard and de PdS invariant at $\eta_1$.  This raises the following questions:
	\begin{enumerate}
		\item Are $|{\rm BD},\eta_1\rangle$ and $|{\rm BD},\eta_0\rangle$ indeed different states? (We will show the answer is in the affirmative.)
		
		\item If they are different, and since $|{\rm BD},\eta_0\rangle$ is the only Hadamard and PdS invariant state at $\eta_0$, $|{\rm BD},\eta_1\rangle$ cannot satisfy both these two properties at $\eta_0$. We will show below that $|{\rm BD},\eta_1\rangle$ is in fact {\em neither} Hadamard nor PdS invariant at instant $\eta_0$. 
		
		\item We will also show that the states $|{\rm BD},\eta_0\rangle$ and $|{\rm BD},\eta_1\rangle$ are connected by time evolution:  $\hat U_{\eta_1,\eta_0}|{\rm BD},\eta_0\rangle=|{\rm BD},\eta_1\rangle$. 
	\end{enumerate}
	Therefore,  we will reach the conclusion that in  Schr\"odinger's picture there exits a one-parameter family of states, $|{\rm BD},\eta\rangle$, each of which is Hadamard and PdS invariant only at time $\eta$, and which are related to each other by time evolution. In Heisenberg's picture, we simply pick one representative in this family of states, and call it {\em the} Bunch-Davies vacuum. 

	In the remaining of this section we prove these statements. 
	\begin{proof}
		To show that $|{\rm BD},\eta_1\rangle$ and $|{\rm BD},\eta_0\rangle$ are different states, we simply need to prove that the phase space elements $\gamma^{\rm BD}_{\vec k}(\eta_1, \vec x)$ do not belong to ${\rm span}\{\gamma^{\rm BD}_{\vec k}(\eta_0, \vec x)\}_{\vec k}$. This can be done by writing $\gamma^{\rm BD}_{\vec k}(\eta_1, \vec x)$ as 
		\be \gamma^{\rm BD}_{\vec k}(\eta_1, \vec x)=\alpha_k\, \gamma^{\rm BD}_{\vec k}(\eta_0, \vec x)+\beta_k\, \gamma^{\rm BD}_{-\vec k}(\eta_0, \vec x)\, ,\ee
		and showing that the coefficients $\beta_k$ are different from zero. The expression for $\beta_k$ are lengthy and not particularly interesting, and we do not write them explicitly. The important information is  that these coefficients are {\em different form zero and $k$-dependent.} This implies that $|{\rm BD},\eta_1\rangle$ is not PdS invariant at $\eta_0$, since we proved above that for all PdS invariant states $\beta_k$ should be $k$-independent.  Furthermore, we find that $\beta_k$ approaches a constant value when $k\to \infty$; this value is different from zero whenever $\eta_1\neq \eta_0$, which shows that  $|{\rm BD},\eta_1\rangle$ is not a Hadamard state at $\eta_0$.
		To prove that $\hat U_{\eta_1,\eta_0}|{\rm BD},\eta_0\rangle=|{\rm BD},\eta_1\rangle$ it suffices to notice that the classical time evolution from $\eta_0$ to $\eta_1$ brings the phase space element $\gamma^{\rm BD}_{\vec k}(\eta_0, \vec x)$ to $\gamma^{\rm BD}_{\vec k}(\eta_1, \vec x)$ ---because the expression in terms of Hankel functions written in  (\ref{BDmodesM}) are exact solutions to the equations of motion. This implies that the evolution of ${\rm span}\{\gamma^{\rm BD}_{\vec k}(\eta_0, \vec x)\}_{\vec k}$ from $\eta_0$ to $\eta_1$  produces ${\rm span}\{\gamma^{\rm BD}_{\vec k}(\eta_1, \vec x)\}_{\vec k}$, and consequently that the Fock vacuum  state $|{\rm BD},\eta_0\rangle$ evolves to $|{\rm BD},\eta_1\rangle$. 
	\end{proof}
	
	That $|{\rm BD},\eta_1\rangle$ is not a Hadamard state at $\eta_0$ is not surprising, since the Hadamard condition involves the form of the spacetime geometry, and the metric tensor of PdS spacetimes changes in time. So, if $|{\rm BD},\eta_1\rangle$ is Hadamard at $\eta_1$, it cannot be at $\eta_0$. But, how to understand that $|{\rm BD},\eta_1\rangle$ is PdS invariant at $\eta_1$ but not at $\eta_0$? The reason is that the Killing vector field $K^{\mu}$ is time dependent. As a result, $K^{\mu}$ does not generate the same transformations at $\eta_0$ and $\eta_1$. In other words, $K^{\mu}$ does not define a unique transformation in the phase space of our field theory, but rather a two-parameter family of  transformations, parameterized by the initial and final times, $\eta_0$ and $\eta_0+\Delta \eta$ (this is equivalent to saying that these transformations are generated by a time-dependent ``Hamiltonian''). The state $|{\rm BD},\eta_1\rangle$ is designed to be invariant under the $K$-flow starting at $\eta_1$, and this makes it not invariant  under the $K$-flow starting at $\eta_0$.

	\bibliography{Refs}

\end{document}